\documentclass[10pt]{article}

\usepackage{amsthm,amsfonts,amsmath,amssymb,mathtools,fullpage}

\usepackage{mathrsfs,wasysym,longtable,pxfonts,txfonts,latexsym,graphicx}

\usepackage{mathptmx,subfigure,cite,float,caption,lipsum}

\usepackage{enumerate,cases,color,subfigure,epsfig,wrapfig,multirow,algorithm,algorithmic}

\usepackage[LY1]{fontenc}

\usepackage[font=bf]{caption}

\DeclareGraphicsExtensions{.pdf,.png,.jpg}

\theoremstyle{definition}

\newtheorem{Thm}{Theorem}

\newtheorem{Rem}{Remark}

\newtheoremstyle{case}{}{}{}{}{}{ \textbf{:} }{ }{}

\theoremstyle{case}

\newtheorem{case}{\textbf{Case}}

\numberwithin{case}{section}

\begin{document}

\title{\textbf{Solitons of nonlinear dispersive wave steered from Navier-Bernoulli hypothesis and Love's hypothesis in the cylindrical elastic rod with compressible Murnaghan's materials}}
\author{Rathinavel Silambarasan$^1$, Adem Kilicman$
^{2,\star}$}
\date{}
\maketitle
\noindent $^{1}$ Department of Information Technology, School of Information Technology and Engineering, Vellore Institute of Technology, Vellore, 632014,
Tamilnadu, India. Email : silambu\_vel@yahoo.co.in\newline
\noindent $^{2}$ Department of Mathematics, Faculty of Science, Universiti Putra Malaysia, 43400 UPM Serdang, Selangor, Malaysia. Email : akilic@upm.edu.my\\
\noindent $^{\star }$ \emph{Corresponding Author}.

\begin{abstract}
The nonlinear dispersive wave equation inside the cylindrical elastic rod is derived by applying the Navier-Bernoulli hypothesis and Love's relation in \cite{5}. The elastic rod is assumed to be composed of the Murnaghan's materials such as Lam$\acute{e}$'s coefficent, Poisson ratio and constitutive constant which are compressible in nature. In this research paper we apply the two integral architectures namely extended sine-Gordon method and modified exponential function method to study the dispersive wave and solved for the solitons and their classifications. The existence of the number of solutions are proved with respect to the linear equation obtained by balancing principle. The related two and three dimensional graphs are simulated and drawn to show the complex structures.
\end{abstract}

\textbf{Keywords} : Dispersive wave, Elastic rod, Murnaghan material, sine-Gordon method, modified exponential function method.\\
\textbf{Mathematical subject classification} : 33E05, 83C15, 35B10.

\section{Introduction}
\label{intr}

The solitons are the non singular waves that are travelling through the any medium without changing their original structure. These solitons are the solitary waves which have the forms of lump (or) kink, bright (or) dark, topological (or) non topological shapes. The waves passing in the solid theory such as inelastic rod, cylindrical rod, hyperelastic rod (or) any other type of rods have the different characterstics based on the materials associated with the rod. These materials are usually Neo-Hookean materials, compressible materials, Murnaghan materials and many other related materials. The waves propagating through any solids are expressed in the form of nonlinear partial differential equations with the coefficients mentioning the constants related the materials of the rod. The problem of solving such nonlinear partial differential equations to study the wave patterns are the many years of research carried out by many people. Usually the aforesaid wave patterns are the solitons, singular waves, peridoic waves, doubly periodic waves and rational waves and any other forms of waves. In the circular cross section of elastic rod the aucoustic waves are expressed in the form of improved Boussinesq equation and solved for the corresponding soliton and the necessary conservation laws in \cite{1}. In the elastic rod deformation waves propagating longitudinally is expressed in the nonlinear partial differential equation and shown the non integrability condition and then solved for the solitons with their interactions in head-on collision is shown in \cite{2}. The propagation of strain waves throught the elastic rod are expressed in the coupled partial partial differential equation and then reduced into single double dispersion equation, then elastic modulli of the rod is taken and related solitary waves are studied in \cite{3}. The extensive class of nonliear waves through the various solids are studied and solved for solitons and also with numerical techniques in \cite{4}. The four forms of nonlinear dispersive equations are derived in the elastic rod in the cylindrical form composed of the materials related to Murnaghan, then the far-field equations for all the four models and it's related solitary waves are given in \cite{5}. The three nonlinear dispersive model equations are,
\begin{enumerate}
\item The first model is systems of five nonlinear equations with five unknowns. (equations (15a), (15b), (15d), (16a) and (16b) in \cite{5}).
\item The second model is coupled nonlinear equation with two unknowns (equations (17) and (18) in \cite{5}).
\item The third model is coupled nonlinear dispersion equation with two unknowns (equations (28) and (29) in \cite{5}).
\end{enumerate}
These three dispersion relations are non integrable and hence do not exist solitary waves (or) solitons. So for all these three models the corresponding far-field equation in the form of KdV equation is derived using reduction perturbation method and solved for the solitary waves in \cite{5}. The fourth model is the nonlinear dispersive waves (equation (34) in \cite{5}) which is derived based on the principle of Navier-Bernoulli hypothesis and Love's hypothesis. The fourth model is integrable and hence the solitary waves exists. In \cite{5}, the nonlinear dispersive equation is derived to make the simplest form of aforementioned three models of disperse relations by using the Navier-Bernoulli hypothesis which states the axial displacement is function of spatial and time variables while the radial displacement is in radial variable. But the Love's hypothesis states that the axial and radial displacements are related by Poisson ratio. Hence in \cite{5} both the aforementioned hypothesis are combined to derive the nonlinear dispersive wave. Therefore the dispersive waves in the elastic rod composed of the compressible Murnaghan's materials are given by the following nonlinear partial differential equation.
\begin{align}
\frac{\partial^2\Phi\left(x, t\right)}{\partial t^2}-\alpha_1\frac{\partial^2\Phi\left(x, t\right)}{\partial x^2}-\frac{n_1^2\delta}{2}\left(\frac{\partial^4\Phi\left(x, t\right)}{\partial t^2\partial x^2}\right)+\frac{n_1^2\delta}{2\beta_1}\left(\frac{\partial^4\Phi\left(x, t\right)}{\partial x^4}\right)+6\alpha_2\epsilon
\left(
\left(\frac{\partial\Phi\left(x, t\right)}{\partial x}\right)^2
+\Phi\left(x, t\right)\left(\frac{\partial^2\Phi\left(x, t\right)}{\partial x^2}\right)
\right)=0.\label{eq-1}
\end{align} 
In the Eq. \eqref{eq-1} $\delta$ and $\epsilon$ are small parameters and
\begin{align}
n_1=\frac{\lambda_1}{2\left(\lambda_1+\mu_1\right)}\ ;\qquad 
\beta_1=\frac{\rho c^2}{\mu_1}\ ;\qquad \alpha_1=\frac{2c_1}{\beta_1\mu_1}\ ;\qquad 
\alpha_2=\frac{c_2}{\beta_1\mu_1}.\label{eq-2}
\end{align}
In the Eq. \eqref{eq-2} $n_1$ is the Poisson ratio $\lambda_1$ and $\mu_1$ are Lam$\acute{e}$'s coefficients and
\begin{align}
c_1=2\left(\lambda_1+\mu_1\right)n_1^2-2\lambda_1n_1+\frac{\lambda_1}{2}+\mu_1\ ;\qquad c_2=-\kappa_1n_1^2+\kappa_3n_1-\kappa_5+\frac{\kappa_6}{n_1}.\label{eq-3}
\end{align}
In the Eq. \eqref{eq-3}
\begin{align}
\kappa_1=2\left(\lambda_1+\mu_1+2\nu_1+\frac{4\nu_2}{3}+\frac{\nu_4}{3}\right)\ ;\qquad \kappa_3=\lambda_1+2\nu_1+4\nu_2.\label{eq-4}
\end{align}
and
\begin{align}
\kappa_5=\lambda_1+2\nu_1+2\nu_2\ ;\qquad \kappa_6=\frac{\lambda_1}{2}+\mu_1+\nu_1+\frac{\nu_2}{3}+\frac{\nu_4}{3}.\label{eq-5}
\end{align}
In the Eqs. \eqref{eq-4} and \eqref{eq-5} $\nu_1\ ,\ \nu_2$ and $\nu_4$ are the constitutive constants. The Eqs. \eqref{eq-2}-\eqref{eq-5} are the compressible Murnaghan materials and the nonlinear dispersive wave $\Phi\left(x, t\right)$ through the cylindrical elastic rod.

The sine-Gordon method is applied to the double sine-Gordon equation, Magma equation and generalized Pochhammer-Chree equation and solved for the doubly periodic solutions in \cite{6}. The coupled Maccari's system is solved by extended sine-Gordon method for travelling wave solutions in \cite{7}. The new systems of Konno-Oone equation is solved by sine-Gordon method and obtained the complex hyperbolic solutions in \cite{8}. The Tzitz$\acute{e}$ica equation, Dodd-Bullough-Mikhailov equation, Tzitz$\acute{e}$ica-Dodd-Bullough equation and Liouville equation are solved using sine-Gordon method in \cite{9}. The conformal time fractional regularized long wave equation (RLW), modified RLW and symmetric RLW are solved using sine-Gordon method in \cite{10}. The extended sinh-Gordon method is applied for the $\left(2+1\right)-$ dimensional hyperbolic Schr$\ddot{o}$dinger equation and cubic-quintic Schr$\ddot{o}$dinger equations for the optical solutions in \cite{11}. The generalized modified Zakharov-Kuznetsov equation and Broer-Kaup-Kupershmidt equation is solved by rational sine-Gordon method in \cite{12}. The Kerr law and quadratic-cubic law nonlinearities of nonlinear Schr$\ddot{o}$dinger equation for optical solitons using sine-Gordon method in \cite{13}. The two coupled nonlinear equations such as variable coefficient nonlinear Schr$\ddot{o}$dinger equation and variable coefficient Davey-Stewartson equation solved for solitary waves by modified sine-Gordon method in \cite{14}.

The modified Kudryashov method applied for generalized Kuramoto-Sivashinsky mequation in \cite{15}. The Kudryashov method, extended Kudryashov method and Riccati equation method applied for Foks-Lennells equation in \cite{16}. The $\left(G^{'}/G^2\right)-$ expansion method and modified Kudryashov method used to solve fractional Zakharov-Kuznetsov equation in \cite{17}. The dual mode Hirota-Satsuma equation is solved using rational sine-cosine method and Kudryashove method in \cite{18}. The fractional Burgers equation solved using generalized Kudryashov method in \cite{19}. The sine-cosine method, simplest equation method, modified Kudryashov method and unified Riccati equation method is applied for nonlinear Schr$\ddot{o}$dinger equation with kerr refractive index having high order dispersions in \cite{20}.

The modified exponential function method (MEFM) is applied for Boussinesq water equation in \cite{21}. The MEFM is applied for longitudinal equation in magneto-electro elastic rod for the complex solutions in \cite{22}. The Phi-four equation is solved for analytical solutions using MEFM in \cite{23}. The coupled long short wave interaction equation is solved for complex structures using MEFM in \cite{24}. The two component second order KdV equation is solved using MEFM for exact solution and finite forward difference method for numerical solution in \cite{25}. The Cahn-Allen equation is solved by MEFM in \cite{26}. Some of the pseudo-parabolic models are solved using MEFM in \cite{27}. The transmission line model is solved by MEFM in \cite{28}. The coupled Miccari's systems solved by MEFM in \cite{29}.

The extended trial equation method (ETEM) is used to solve Biswas-Milovic equation for solitons in \cite{30}. The fractional Schr$\ddot{o}$dinger equation with perturbation terms is solved by extended trial function method (ETFM) in \cite{31}. The ETFM is applied for Kundu-Eckhaus equation in birefringent fibers for the solitons and conservation laws in \cite{32}. The Kundu-Mukherjee-Naskar equation is solved using ETFM in \cite{33}. The Biswas-Arshed equation is solved by ETFM in \cite{34} for optical solitons. The celeberated sine-cosine method is applied for certain nonlinear wave equations in \cite{35}. The variable coefficient Schr$\ddot{o}$dinger equation is solved by generalized extended tanh function method, sine-cosine method and exp function method in \cite{36}. The generalized KdV equation and the modified KdV equation is solved sine-cosine method in \cite{37}. The sine-cosine method and Bernoulli's equation method is used to solved twin core couplers with Kerr law, power law, parabolic law and dual power law in \cite{38}. The system of equal width equation is solved by sine-cosine method in \cite{39}. The double dispersion equation (DDE) in the Murnaghan rod is solved for Jacobi solutions using F expansion method in \cite{40}. The non dissipated DDE in the micro structured solids is solved for the periodic waves using F expansion method in \cite{42}. The Eq. \eqref{eq-1} is solved for the doubly periodic solutions using F expansion method in \cite{42}.

In this present paper we contribute to the following.
\begin{enumerate}
\item The nonlinear dispersive equation in the elastic rod composed of Murnaghan materials given in the Eq. \eqref{eq-1} is solved by extended sine-Gordon equation expansion method for the optical solitons and their classifications.
\item The modified exponential function method is applied secondly to the Eq. \eqref{eq-1} and obtain the another set of solitons.
\item The two and three dimensional plottings are given for the selective unknown function $\Phi\left(x, t\right)$ solitons.
\end{enumerate}

\section{Description of the integral architectures}
\label{desc}

Let the given nonlinear partial differential equation in the space and time variable be expressed in the following polynomial form.
\begin{align}
P\left(\Phi_{tt}+\Phi_{xx}+\Phi_{xt}+\cdots\right)=0.\label{eq-6}
\end{align}
In the Eq. \eqref{eq-6} $\Phi=\Phi\left(x, t\right)$ and the subscripts represents the partial derivatives. Suppose $\Phi=\Phi\left(x, t\right)=u\left(\xi\right)$ with $\xi$ being $\mu\left(x-\lambda t\right)$ with $\mu$ and $\lambda$ are the wave number and frequency respectively. Then the Eq. \eqref{eq-6} reduces into the following polynomial form of nonlinear ordinary differential equation.
\begin{align}
O\left(-\lambda\mu u^{'}+\lambda^2\mu^2u^{''}+\mu^3u^{'''}+\cdots\right)=0.\label{eq-7}
\end{align}
In the Eq. \eqref{eq-7} $u=u\left(\xi\right)$ and the superscripts represents the derivative with respect to $\xi$.

\subsection{Extended sine Gordon Method}
\label{esgm}

Consider the sine-Gordon equation of the following form \cite{6,7,8,9,10,11,12,13,14}.
\begin{align}
u_{xx}-u_{tt}=m^2\sin\left(u\right).\label{eq-8}
\end{align}
In the Eq. \eqref{eq-8} $u=u\left(x, t\right)$ and $m$ is the non-zero constant. Suppose $u\left(x, t\right)=U\left(\xi\right)$ with $\xi=kx+lt$. So the Eq. \eqref{eq-8} leads to the following ordinary differential equation.
\begin{align}
U^{''}=\frac{m^2}{k^2-l^2}\sin\left(U\right).\label{eq-9}
\end{align}
Next multiplying the both sides of Eq. \eqref{eq-9} by $U^{'}$ and then integrating gives the following differential equation.
\begin{align}
\left(\left(\frac{U}{2}\right)^{'}\right)^2=
\frac{m^2}{k^2-l^2}\sin^2\left(\frac{U}{2}\right)+K.\label{eq-10}
\end{align}
In the Eq. \eqref{eq-10} $K$ is the integration. Next in the Eq. \eqref{eq-9} taking $K=0$, $\frac{U}{2}=w\left(\xi\right)$ and $\frac{m^2}{k^2-l^2}=a^2$ leads to the following differential equation.
\begin{align}
w^{'}=a\sin\left(w\right).\label{eq-11}
\end{align}
When $a=1$ in the Eq. \eqref{eq-11} leads to the differential equation.
\begin{align}
w^{'}=\sin\left(w\right).\label{eq-12}
\end{align}
The Eq. \eqref{eq-12} is the reduced form of sine-Gordon equation given in the Eq. \eqref{eq-8}. The solution of Eq. \eqref{eq-12} is given by the following equations.
\begin{align}
\sin\left(w\right)&=\sin\left(w\left(\xi\right)\right)
=\frac{2p\exp\left(\xi\right)}{p^2\exp\left(2\xi\right)+1}|_{p=1}
=\mbox{sech}\left(\xi\right)=i\mbox{csch}\left(\xi\right)\ ;\ i=\sqrt{-1}.\label{eq-13}\\
\cos\left(w\right)&=\cos\left(w\left(\xi\right)\right)
=\frac{p^2\exp\left(2\xi\right)-1}{p^2\exp\left(2\xi\right)+1}|_{p=1}
=\tanh\left(\xi\right)=\coth\left(\xi\right).\label{eq-14}
\end{align}
In the Eqs. \eqref{eq-13} and \eqref{eq-14} $p$ is the non-zero integration constant. Hence the solution of Eq. \eqref{eq-7} is assumed in the following form.
\begin{align}
u\left(\xi\right)&=A_0+\sum_{i=1}^N\tanh^{i-1}\left(\xi\right)
\left[A_i\tanh\left(\xi\right)+B_i\mbox{sech}\left(\xi\right)\right].\label{eq-15}\\
u\left(\xi\right)&=A_0+\sum_{i=1}^N\coth^{i-1}\left(\xi\right)
\left[A_i\coth\left(\xi\right)+B_ii\mbox{csch}\left(\xi\right)\right]\ ;\ i=\sqrt{-1}.\label{eq-16}
\end{align}
Now by using the Eqs. \eqref{eq-13} and \eqref{eq-14}, Eqs. \eqref{eq-15} and \eqref{eq-16} is written in the following form.
\begin{align}
u\left(\omega\right)=A_0+\sum_{i=1}^N\cos^{i-1}\left(\omega\right)
\left[A_i\cos\left(\omega\right)+B_i\sin\left(\omega\right)\right].\label{eq-17}
\end{align}
In the Eqs. \eqref{eq-15}, \eqref{eq-16} and \eqref{eq-17} $N$ is the positive integer calculated from the Eq. \eqref{eq-7} by balancing principle. Hence substituting Eq. \eqref{eq-17} into the Eq. \eqref{eq-7} leads to the algebraic systems of equations in trigonometric functions. Solving the coefficents of trigonometric functions gives the unknowns of Eqs. \eqref{eq-15} and \eqref{eq-16} substituting them in the Eqs. \eqref{eq-15} and \eqref{eq-16} gives the exact solutions of Eq. \eqref{eq-6}.

\subsection{Modified exponential function method}
\label{mefm}

Let assume the solution of Eq. \eqref{eq-7} in the following form.
\begin{align}
u\left(\xi\right)=\frac{\sum_{i=0}^NP_i\left[\exp\left(-\varphi\left(\xi\right)\right)\right]^i}
{\sum_{j=0}^MQ_j\left[\exp\left(-\varphi\left(\xi\right)\right)\right]^j}.\label{eq-18}
\end{align}
In the Eq. \eqref{eq-18} $P_i$ and $Q_j$ are the constants to be computed, $N$ and $M$ are computed using balancing principle and $\varphi\left(\xi\right)$ is the solution of the following differential equation \cite{21,22,23,24,25,26,27,28,29}.
\begin{align}
\frac{d\varphi\left(\xi\right)}{d\xi}=\exp\left(-\varphi\left(\xi\right)\right)
+\sigma\exp\left(\varphi\left(\xi\right)\right)+\tau.\label{eq-19}
\end{align}
Eq. \eqref{eq-19} has the following five set of solutions.

\textbf{Set 1}. When $\sigma\ne 0$ and $\tau^2-4\sigma>0$.
\begin{align}
\varphi\left(\xi\right)=\ln\left[-\frac{\sqrt{\tau^2-4\sigma}}{2\sigma}
\tanh\left(\frac{\sqrt{\tau^2-4\sigma}}{\sigma}\left(\xi+e\right)\right)
-\frac{\tau}{2\sigma}\right].\label{eq-20}
\end{align}
\textbf{Set 2}. When $\sigma\ne 0$ and $\tau^2-4\sigma<0$.
\begin{align}
\varphi\left(\xi\right)=\ln\left[\frac{\sqrt{-\tau^2+4\sigma}}{2\sigma}
\tan\left(\frac{\sqrt{-\tau^2+4\sigma}}{\sigma}\left(\xi+e\right)\right)
-\frac{\tau}{2\sigma}\right].\label{eq-21}
\end{align}
\textbf{Set 3}. When $\sigma=0$, $\tau\ne 0$ and $\tau^2-4\sigma>0$.
\begin{align}
\varphi\left(\xi\right)=-\ln\left[\frac{\tau}{\exp\left(\tau\left(\xi+e\right)\right)-1}\right].\label{eq-22}
\end{align}
\textbf{Set 4}. When $\sigma\ne 0$, $\tau\ne 0$ and $\tau^2-4\sigma=0$.
\begin{align}
\varphi\left(\xi\right)=\ln\left[-\frac{2\tau\left(\xi+e\right)+4}{\tau^2\left(\xi+e\right)}\right].\label{eq-23}
\end{align}
\textbf{Set 5}. When $\sigma=0$, $\tau=0$ and $\tau^2-4\sigma=0$.
\begin{align}
\varphi\left(\xi\right)=\ln\left[\xi+e\right].\label{eq-24}
\end{align}
In the Eqs. \eqref{eq-20}-\eqref{eq-24} $e$ is the integration constant. Now substituting Eqs. \eqref{eq-18} and \eqref{eq-19} into the Eq. \eqref{eq-7} gives the polynomial in $\exp\left(-\varphi\left(\xi\right)\right)$. Extracting the coefficent of $\exp\left(-\varphi\left(\xi\right)\right)$ and equating to zero and then solving gives the constants $P_i$, $Q_j$, $\lambda$ and $\mu$ hence the exact solution of Eq. \eqref{eq-6} is obtained.

\section{Analysis of nonlinear dispersive wave}
\label{anly}

To convert the dispersive wave equation given in the Eq. \eqref{eq-1} the wave transform is applied $\Phi\left(x, t\right)=u\left(\xi\right)$ in which $\xi=\mu\left(x-\lambda t\right)$ to the Eq. \eqref{eq-1}. So the nonlinear partial differential equation is converted into the ordinary differential equation.
\begin{align}
\frac{n_1^2\mu^2\delta}{2}\left(\frac{1}{\beta_1}-\lambda^2\right)
\left(\frac{d^2u\left(\xi\right)}{d\xi^2}\right)+3\alpha_2\epsilon u^2\left(\xi\right)+\left(\lambda^2-\alpha_1\right)u\left(\xi\right)=0.\label{eq-25}
\end{align}
The Eq. \eqref{eq-25} is the working model for the rest of the paper. First the extended sine-Gordon equation expansion method is applied to the Eq. \eqref{eq-25} and then the modified exponential function method is applied to the Eq. \eqref{eq-25} to obtain the solitons of the dispersive wave equation \eqref{eq-1}.

\section{Solitons obtained using extended sine Gordon method}
\label{soli1}

\begin{figure}[ht]
\centering
\par
\mbox{\subfigure[$\mathscr{R}\left(\Phi\left(x, t\right)\right)$ in Eq. \eqref{eq-43}]{\includegraphics[height=3in,width=3in]{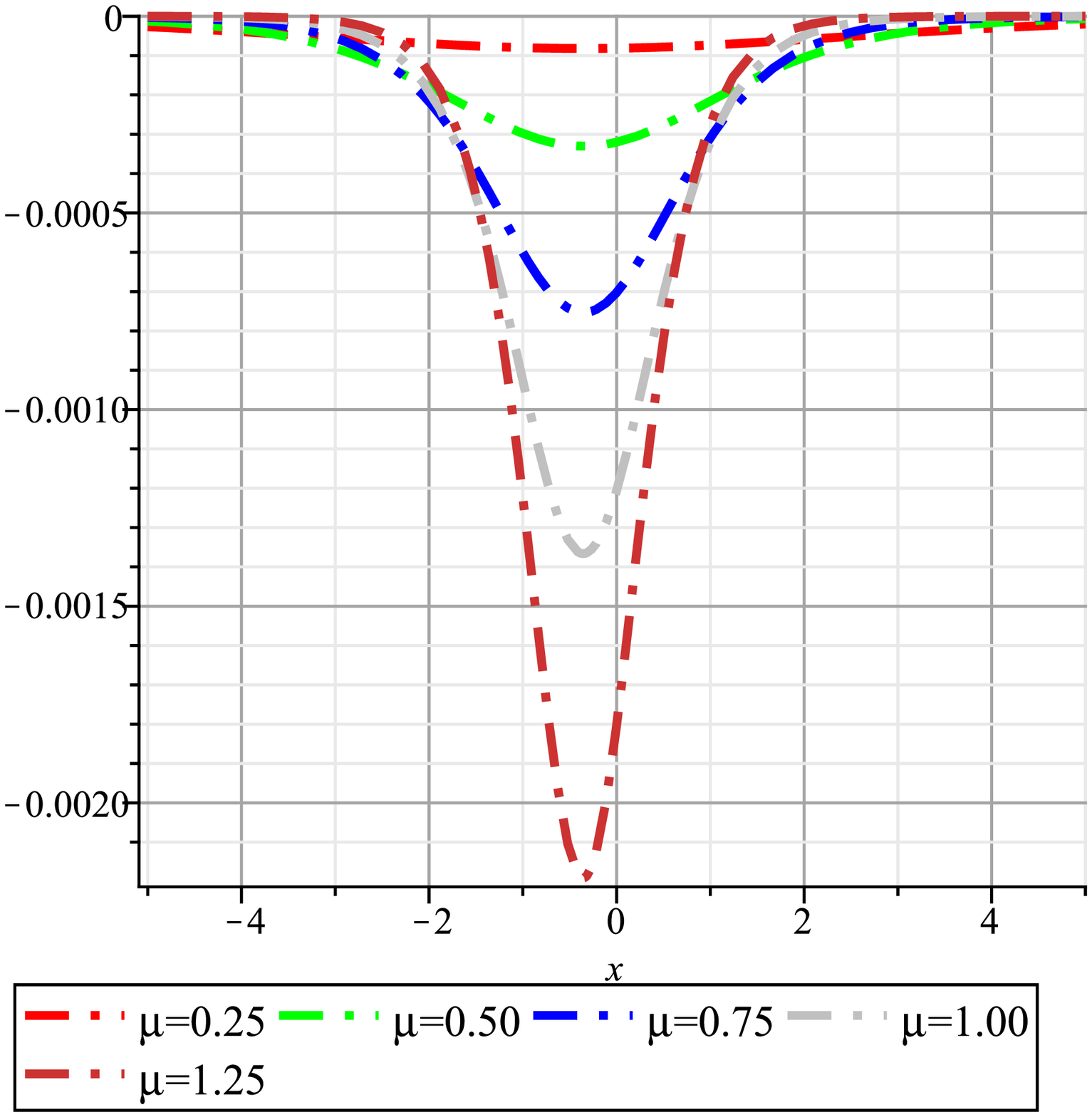}}\quad

\subfigure[$\mathscr{I}\left(\Phi\left(x, t\right)\right)$ in Eq. \eqref{eq-43}]{\includegraphics[height=3in,width=3in]{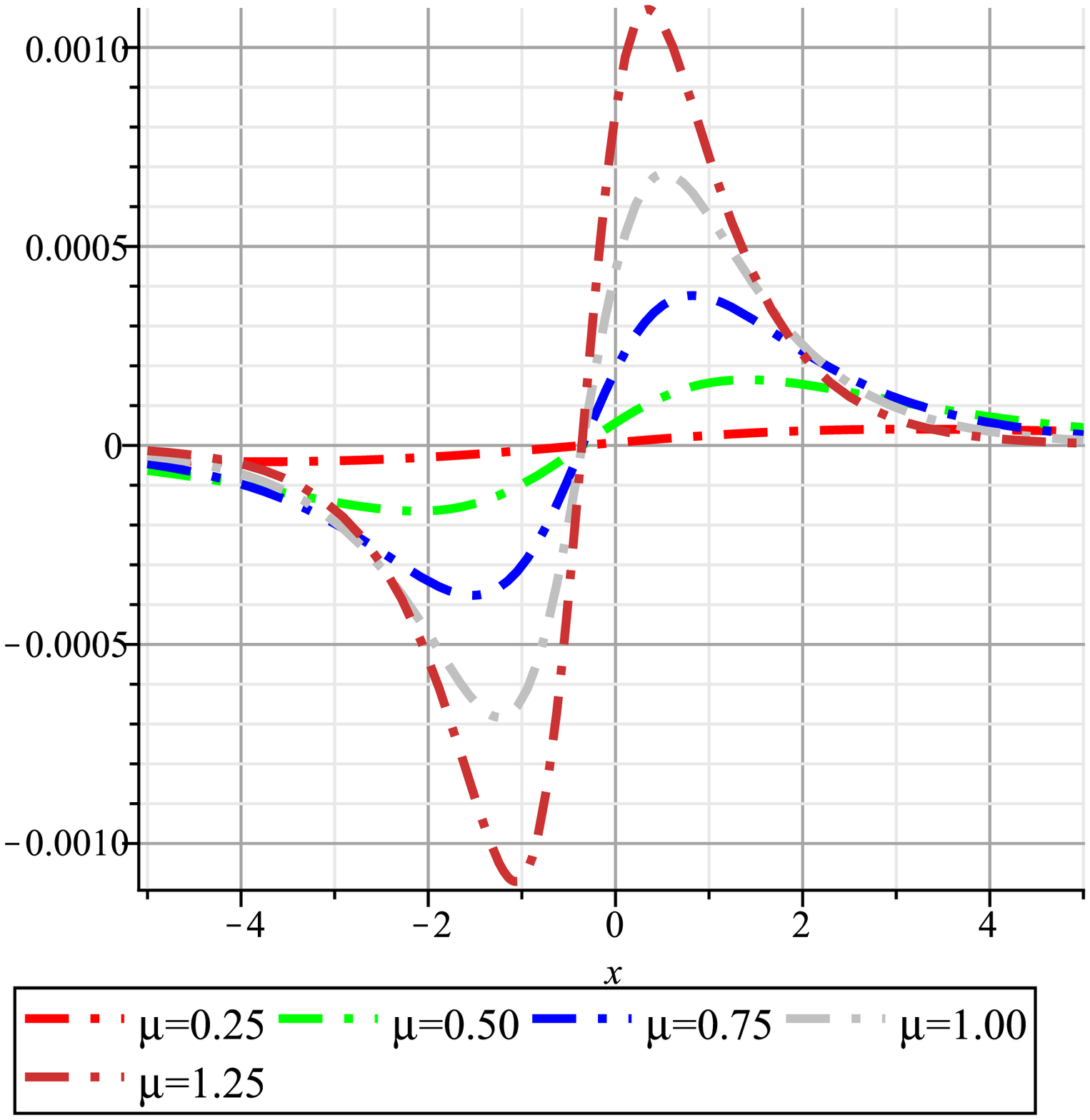} }}
\caption{\textbf{2D graph of the real and imaginary part of the soliton $\Phi\left(x, t\right)$ given in the Eq. \eqref{eq-43} in the domain $x\in[-5\ ,\ 5]$.}}
\label{fig-1}
\end{figure}

\begin{figure}[ht]
\centering
\par
\mbox{\subfigure[$\mbox{Positive value of}\ \mathscr{R}\left(\Phi\left(x, t\right)\right)$ in Eq. \eqref{eq-48}]{\includegraphics[height=3in,width=3in]{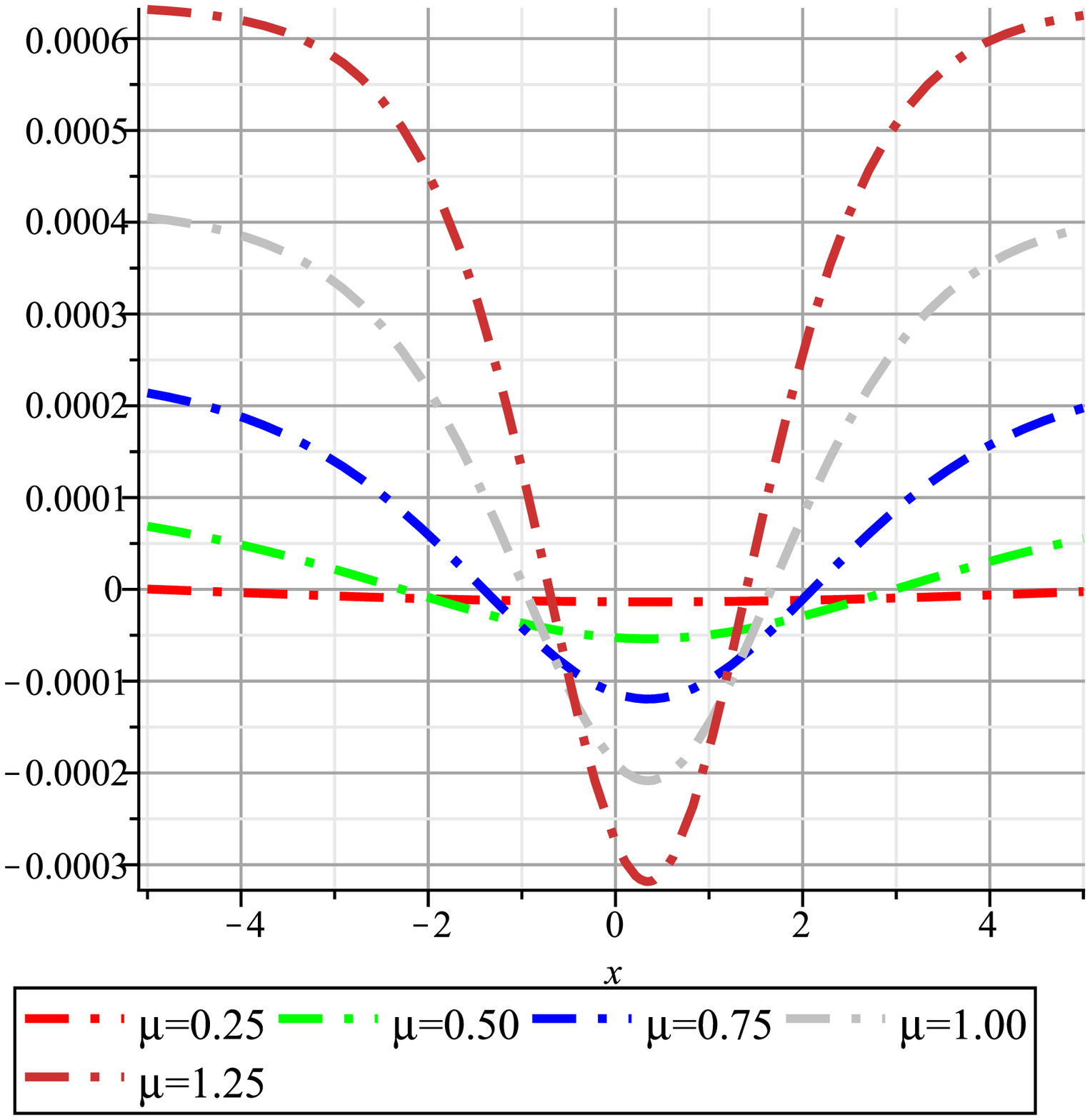}}\quad

\subfigure[$\mbox{Negative value of}\ \mathscr{R}\left(\Phi\left(x, t\right)\right)$ in Eq. \eqref{eq-48}]{\includegraphics[height=3in,width=3in]{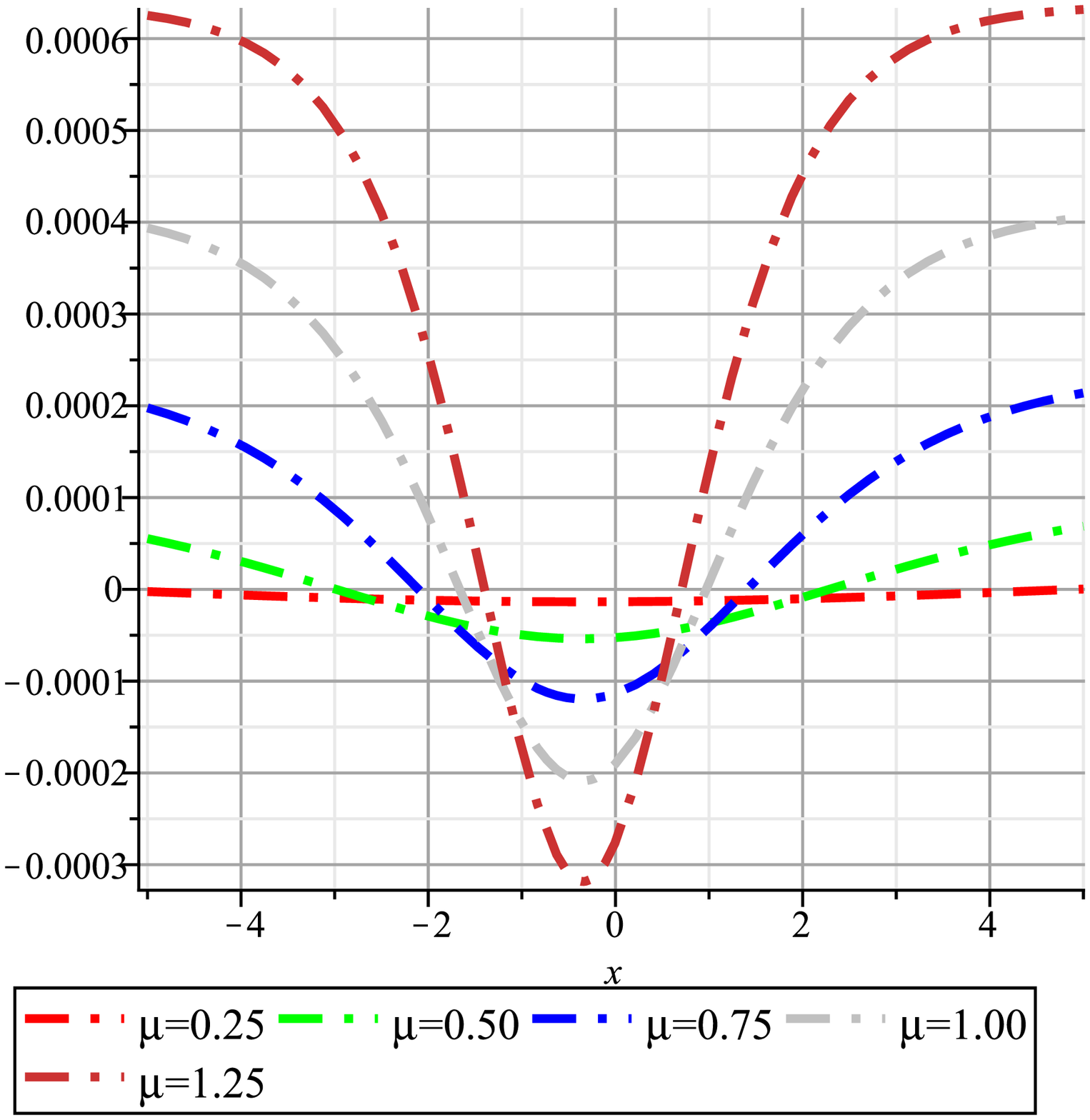} }}
\caption{\textbf{2D graph of the real part of the positive and negative $\xi_2$ in the soliton $\Phi\left(x, t\right)$ given in the Eq. \eqref{eq-48} in the domain $x\in[-5\ ,\ 5]$.}}
\label{fig-2}
\end{figure}

\begin{figure}[ht]
\centering
\par
\mbox{\subfigure[$\mathscr{R}\left(\Phi\left(x, t\right)\right)$ in Eq. \eqref{eq-35}]{\includegraphics[height=3in,width=3in]{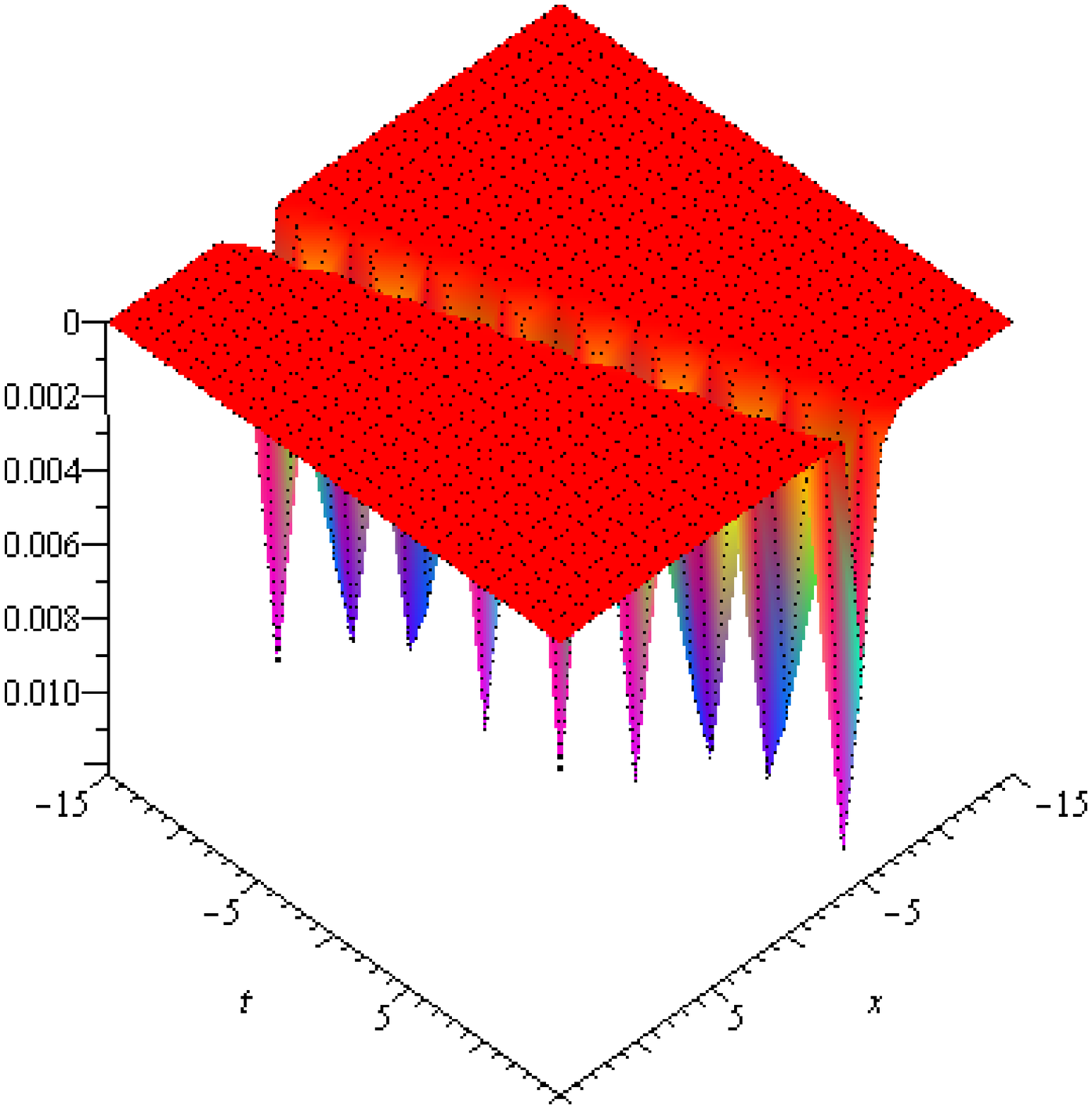}}\quad

\subfigure[$\mathscr{R}\left(\Phi\left(x, t\right)\right)$ in Eq. \eqref{eq-36}]{\includegraphics[height=3in,width=3in]{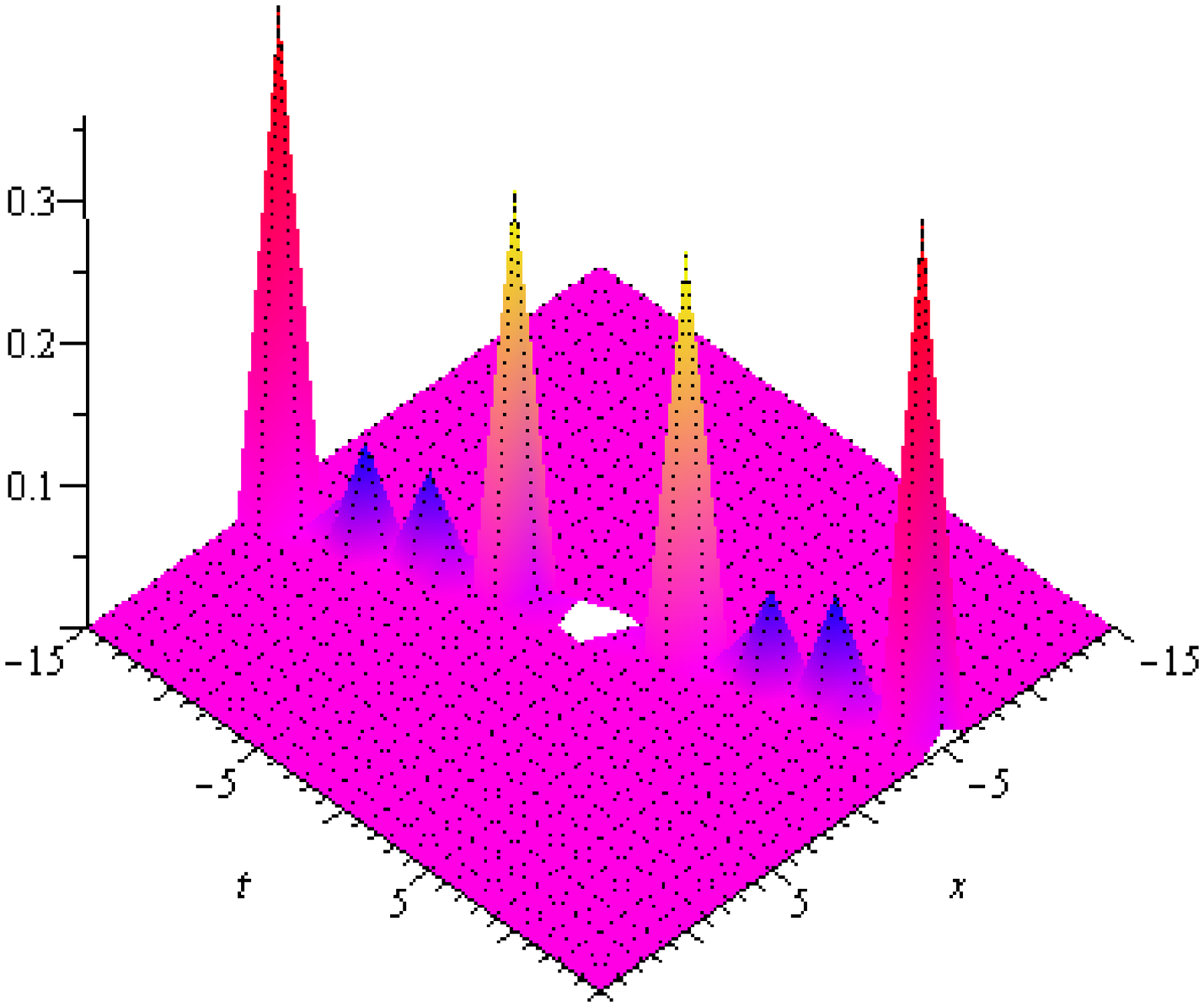} }}
\caption{\textbf{3D graph of the real part of the solitons $\Phi\left(x, t\right)$ given in the Eqs. \eqref{eq-35} and \eqref{eq-36} in the domain $x\ ,\ t\in[-15\ ,\ 15]$.}}
\label{fig-3}
\end{figure}

\begin{figure}[ht]
\centering
\par
\mbox{\subfigure[$\mathscr{R}\left(\Phi\left(x, t\right)\right)$ in Eq. \eqref{eq-51}]{\includegraphics[height=3in,width=3in]{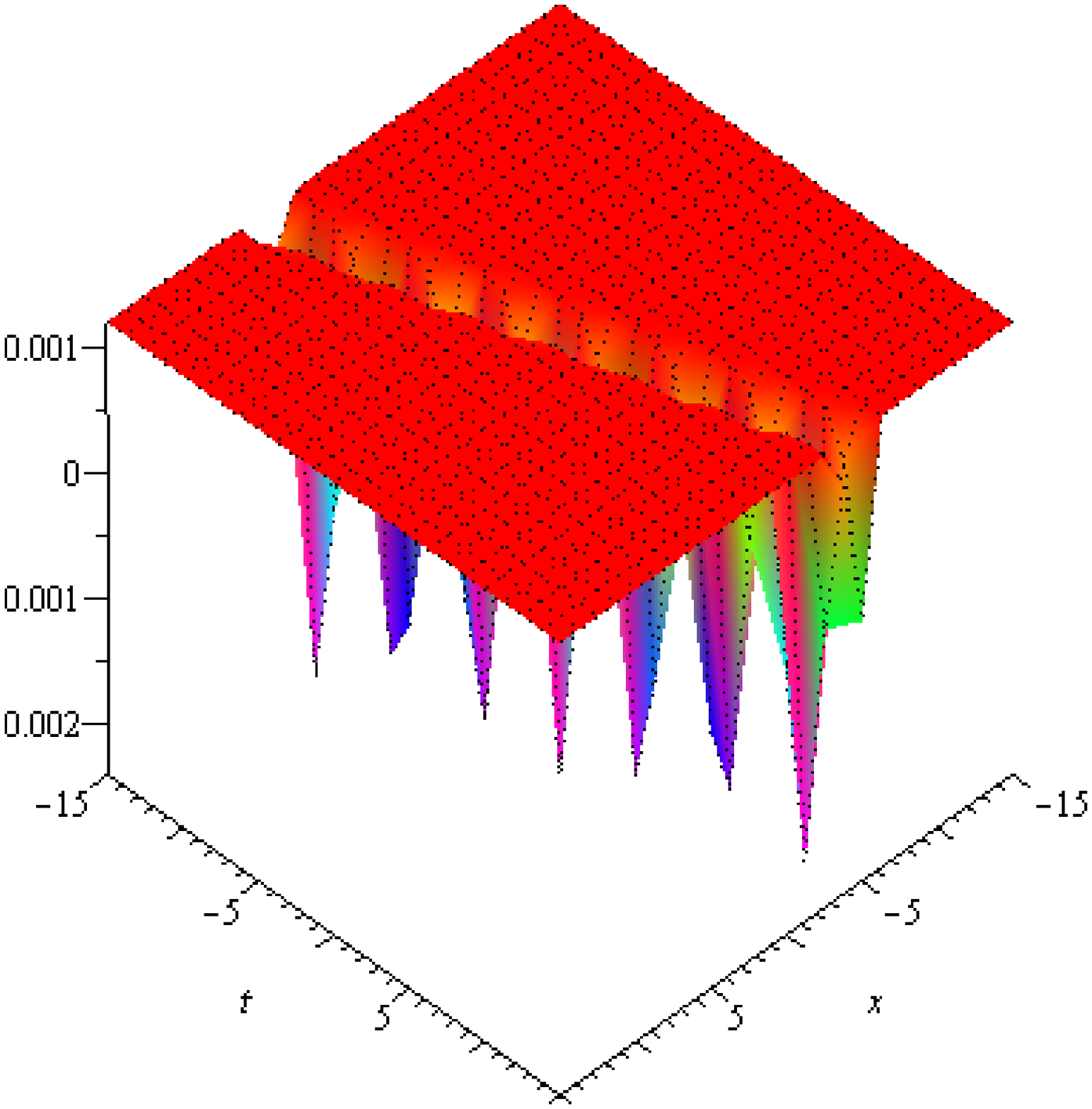}}\quad

\subfigure[$\mathscr{I}\left(\Phi\left(x, t\right)\right)$ in Eq. \eqref{eq-51}]{\includegraphics[height=3in,width=3in]{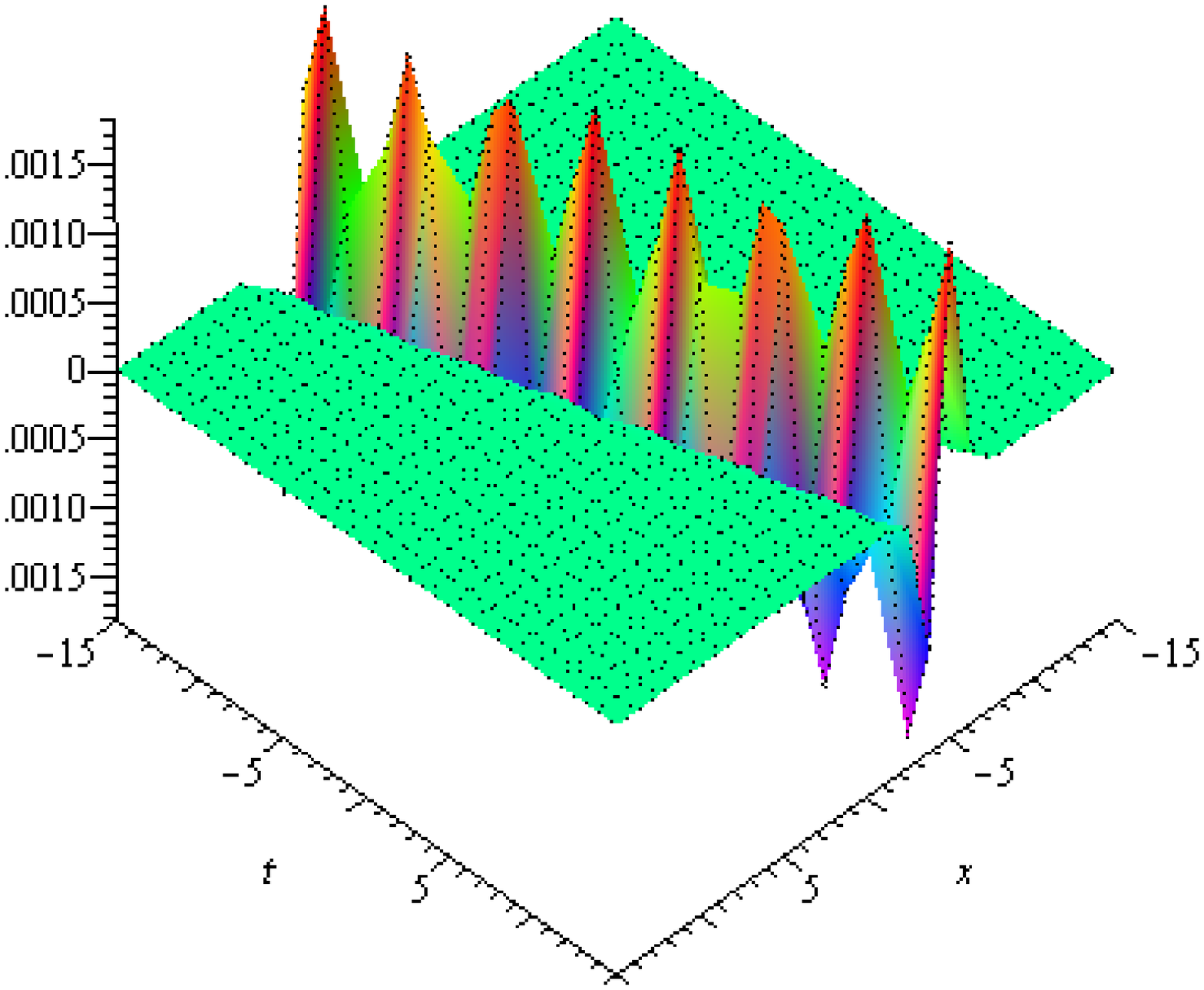} }}
\caption{\textbf{3D graph of the real and imaginary part of the soliton $\Phi\left(x, t\right)$ given in the Eq. \eqref{eq-51} in the domain $x\ ,\ t\in[-15\ ,\ 15]$.}}
\label{fig-4}
\end{figure}

The comparision of $\frac{d^2u\left(\xi\right)}{d\xi^2}$ and $u^2\left(\xi\right)$ in the Eq. \eqref{eq-25} gives $N=2$. Hence let the initial solution of Eq. \eqref{eq-25} is assumed in the following form.
\begin{align}
u\left(\xi\right)&=A_0+\sum_{i=1}^2\tanh^{i-1}\left(\xi\right)
\left[A_i\tanh\left(\xi\right)+B_i\mbox{sech}\left(\xi\right)\right].\label{eq-26}\\
u\left(\xi\right)&=A_0+\sum_{i=1}^2\coth^{i-1}\left(\xi\right)
\left[A_i\coth\left(\xi\right)+B_ii\mbox{csch}\left(\xi\right)\right]\ ;\ i=\sqrt{-1}.\label{eq-27}
\end{align}
From the Eqs. \eqref{eq-13} and \eqref{eq-14} we also have the following from the Eq. \eqref{eq-17}.
\begin{align}
u\left(\omega\right)=A_0+\sum_{i=1}^2\cos^{i-1}\left(\omega\right)
\left[A_i\cos\left(\omega\right)+B_i\sin\left(\omega\right)\right].\label{eq-28}
\end{align}
Now substituting Eq. \eqref{eq-26} in the Eq. \eqref{eq-25} and using Eq. \eqref{eq-17} yields the polynomial in $\sin\left(\omega\right)$, $\cos\left(\omega\right)$ and their powers. Next extracting the coefficent of each of the $\sin\left(\omega\right)$ and it's powers and then equating each of them to zero results in the nine systems of overdetrmined algebraic equations. Solving the aforesaid systems gives the unknowns $A_0\ ,\ A_1\ ,\ A_2\ ,\ B_1\ ,\ B_2\ ,\ \lambda$ and $\mu$. Next substituting these coefficents in the Eqs. \eqref{eq-26} and \eqref{eq-27} gives the exact solitons of the dispersive wave equation Eq. \eqref{eq-1}. These solitons are separately reported in the following cases.
\begin{case}
\label{sol1}
When
\begin{align}
A_0=-\frac{\delta n_1^2\mu^2\left(\alpha_1\beta_1-1\right)}{3\alpha_2\beta_1\epsilon\left(1+2\delta n_1^2\mu^2\right)}\ ;\ A_1=0\ ;\ A_2=\frac{\delta n_1^2\mu^2\left(\alpha_1\beta_1-1\right)}{\alpha_2\beta_1\epsilon\left(1+2\delta n_1^2\mu^2\right)}.\label{eq-29}
\end{align}
\begin{align}
B_1=B_2=0\ ;\ \lambda=\pm\sqrt{-\frac{-\alpha_1\beta_1-2\delta n_1^2\mu^2}{\beta_1\left(1+2\delta n_1^2\mu^2\right)}}\ ;\ \mu=\mu.\label{eq-30}
\end{align}
Substituting $A_0$ and $A_2$ from Eq. \eqref{eq-29}, $\lambda$ from Eq. \eqref{eq-30} in the Eqs. \eqref{eq-26} and \eqref{eq-27} pertains the topological soliton of Eq. \eqref{eq-1}.
\begin{align}
\Phi\left(x, t\right)=\frac{\delta n_1^2\mu^2\left(\alpha_1\beta_1-1\right)}{6\alpha_2\beta_1\epsilon\left(\frac{1}{2}+\delta n_1^2\mu^2\right)}
\left\{
-1+\tanh^2\left(\mu\left[x\mp\sqrt{\frac{\alpha_1\beta_1+2\delta n_1^2\mu^2}{\beta_1\left(1+2\delta n_1^2\mu^2\right)}}t\right]\right)
\right\}.\label{eq-31}
\end{align}
and singular soliton of Eq. \eqref{eq-1}.
\begin{align}
\Phi\left(x, t\right)=\frac{\delta n_1^2\mu^2\left(\alpha_1\beta_1-1\right)}{6\alpha_2\beta_1\epsilon\left(\frac{1}{2}+\delta n_1^2\mu^2\right)}
\left\{
-1+\coth^2\left(\mu\left[x\mp\sqrt{\frac{\alpha_1\beta_1+2\delta n_1^2\mu^2}{\beta_1\left(1+2\delta n_1^2\mu^2\right)}}t\right]\right)
\right\}.\label{eq-32}
\end{align}
\end{case}

\begin{case}
\label{sol2}
When
\begin{align}
A_0=\frac{\delta n_1^2\mu^2\left(\alpha_1\beta_1-1\right)}{\alpha_2\beta_1\epsilon\left(-1+2\delta n_1^2\mu^2\right)}\ ;\ A_1=0\ ;\ A_2=-\frac{\delta n_1^2\mu^2\left(\alpha_1\beta_1-1\right)}{\alpha_2\beta_1\epsilon\left(-1+2\delta n_1^2\mu^2\right)}.\label{eq-33}
\end{align}
\begin{align}
B_1=B_2=0\ ;\ \lambda=\pm\sqrt{-\frac{\alpha_1\beta_1-2\delta n_1^2\mu^2}{-\beta_1\left(1-2\delta n_1^2\mu^2\right)}}\ ;\ \mu=\mu.\label{eq-34}
\end{align}
Substituting $A_0$ and $A_2$ from the Eq. \eqref{eq-33} and $\lambda$ from Eq. \eqref{eq-34} into the Eqs. \eqref{eq-26} and \eqref{eq-27} yields the topological soliton of Eq. \eqref{eq-1}. The 3D plot for the real part of the soliton is given in the Figure \ref{fig-3}.
\begin{align}
\Phi\left(x, t\right)=-\frac{\delta n_1^2\mu^2\left(\alpha_1\beta_1-1\right)}{\alpha_2\beta_1\epsilon\left(-1+2\delta n_1^2\mu^2\right)}
\left\{
-1+\tanh^2\left(\mu\left[x\mp\sqrt{\frac{-\alpha_1\beta_1+2\delta n_1^2\mu^2}{\beta_1\left(-1+2\delta n_1^2\mu^2\right)}}t\right]\right)
\right\}.\label{eq-35}
\end{align}
and singular soliton of Eq. \eqref{eq-1}. The 3D plot for the real part of the soliton is given in the Figure \ref{fig-3}.
\begin{align}
\Phi\left(x, t\right)=-\frac{\delta n_1^2\mu^2\left(\alpha_1\beta_1-1\right)}{\alpha_2\beta_1\epsilon\left(-1+2\delta n_1^2\mu^2\right)}
\left\{
-1+\coth^2\left(\mu\left[x\mp\sqrt{\frac{-\alpha_1\beta_1+2\delta n_1^2\mu^2}{\beta_1\left(-1+2\delta n_1^2\mu^2\right)}}t\right]\right)
\right\}.\label{eq-36}
\end{align}
\end{case}

\begin{case}
\label{sol3}
When
\begin{align}
A_0=\frac{\delta n_1^2\mu^2\left(\alpha_1\beta_1-1\right)}{\alpha_2\beta_1\epsilon\left(\delta n_1^2\mu^2-2\right)}\ ;\ A_1=0\ ;\ A_2=-\frac{\delta n_1^2\mu^2\left(\alpha_1\beta_1-1\right)}{\alpha_2\beta_1\epsilon\left(\delta n_1^2\mu^2-2\right)}.\label{eq-37}
\end{align}
\begin{align}
B_1=0\ ;\ B_2=\frac{i\delta n_1^2\mu^2\left(\alpha_1\beta_1-1\right)}{\alpha_2\beta_1\epsilon\left(\delta n_1^2\mu^2-2\right)}\ ,\ i=\sqrt{-1}\ ;\ \lambda=\pm\sqrt{-\frac{2\alpha_1\beta_1-\delta n_1^2\mu^2}{\beta_1\left(-2+\delta n_1^2\mu^2\right)}}\ ;\ \mu=\mu.\label{eq-38}
\end{align}
Now substituting $A_0$ and $A_2$ from the Eq. \eqref{eq-37}, $B_2$ and $\lambda$ from the Eq. \eqref{eq-38} into the Eqs. \eqref{eq-26} and \eqref{eq-27} leads to the compound topological-non-topological soliton of Eq. \eqref{eq-1}.
\begin{align}
\Phi\left(x, t\right)=\frac{\delta n_1^2\mu^2}{\alpha_2\beta_1\epsilon}
\left\{
\frac{\left(\alpha_1\beta_1-1\right)
\left[-\tanh^2\left(\xi_1\right)+i\tanh\left(\xi_1\right)\mbox{sech}\left(\xi_1\right)+1\right]}
{\left(-2+\delta n_1^2\mu^2\right)}
\right\}\ ;\ i=\sqrt{-1}.\label{eq-39}
\end{align}
and compound singular soliton of Eq. \eqref{eq-1}
\begin{align}
\Phi\left(x, t\right)=-\frac{\delta n_1^2\mu^2}{\alpha_2\beta_1\epsilon}
\left\{
\frac{\left(\alpha_1\beta_1-1\right)
\left[-\coth^2\left(\xi_1\right)+\coth\left(\xi_1\right)\mbox{csch}\left(\xi_1\right)+1\right]}
{\left(-2+\delta n_1^2\mu^2\right)}
\right\}.\label{eq-40}
\end{align}
In the Eqs. \eqref{eq-39} and \eqref{eq-40} $\xi_1=\left(\mu\left[x\mp\sqrt{\frac{-2\alpha_1\beta_1+\delta n_1^2\mu^2}{\beta_1\left(-2+\delta n_1^2\mu^2\right)}}t\right]\right)$.
\end{case}

\begin{case}
\label{sol4}
When
\begin{align}
A_0=\frac{\delta n_1^2\mu^2\left(\alpha_1\beta_1-1\right)}{\alpha_2\beta_1\epsilon\left(\delta n_1^2\mu^2-2\right)}\ ;\ A_1=0\ ;\ A_2=-\frac{\delta n_1^2\mu^2\left(\alpha_1\beta_1-1\right)}{\alpha_2\beta_1\epsilon\left(\delta n_1^2\mu^2-2\right)}.\label{eq-41}
\end{align}
\begin{align}
B_1=0\ ;\ B_2=-\frac{i\delta n_1^2\mu^2\left(\alpha_1\beta_1-1\right)}{\alpha_2\beta_1\epsilon\left(\delta n_1^2\mu^2-2\right)}\ ,\ i=\sqrt{-1}\ ;\ \lambda=\pm\sqrt{-\frac{2\alpha_1\beta_1-\delta n_1^2\mu^2}{\beta_1\left(-2+\delta n_1^2\mu^2\right)}}\ ;\ \mu=\mu.\label{eq-42}
\end{align}
Next substituting $A_0$ and $A_2$ from the Eq. \eqref{eq-41}, $B_2$ and $\lambda$ from the Eq. \eqref{eq-42} into the Eqs. \eqref{eq-26} and \eqref{eq-27} yields the compound topological-non-topological soliton of Eq. \eqref{eq-1}. The 2D plots for the real and imaginary part of the soliton is given in the Figure \ref{fig-1}.
\begin{align}
\Phi\left(x, t\right)=-\frac{\delta n_1^2\mu^2}{\alpha_2\beta_1\epsilon}
\left\{
\frac{\left(\alpha_1\beta_1-1\right)
\left[\tanh^2\left(\xi_1\right)+i\tanh\left(\xi_1\right)\mbox{sech}\left(\xi_1\right)-1\right]}
{\left(-2+\delta n_1^2\mu^2\right)}
\right\}\ ;\ i=\sqrt{-1}.\label{eq-43}
\end{align}
and compound singular soliton of Eq. \eqref{eq-1}
\begin{align}
\Phi\left(x, t\right)=-\frac{\delta n_1^2\mu^2}{\alpha_2\beta_1\epsilon}
\left\{
\frac{\left(\alpha_1\beta_1-1\right)
\left[\coth^2\left(\xi_1\right)-\coth\left(\xi_1\right)\mbox{csch}\left(\xi_1\right)-1\right]}
{\left(-2+\delta n_1^2\mu^2\right)}
\right\}.\label{eq-44}
\end{align}
In the Eqs. \eqref{eq-43} and \eqref{eq-44} $\xi_1$ is given in the Eq. \eqref{eq-39}.
\end{case}

\begin{case}
\label{sol5}
When
\begin{align}
A_0=-\frac{2\delta n_1^2\mu^2\left(\alpha_1\beta_1-1\right)}{3\alpha_2\beta_1\epsilon\left(\delta n_1^2\mu^2+2\right)}\ ;\ A_1=0\ ;\ A_2=\frac{\delta n_1^2\mu^2\left(\alpha_1\beta_1-1\right)}{\alpha_2\beta_1\epsilon\left(\delta n_1^2\mu^2+2\right)}.\label{eq-45}
\end{align}
\begin{align}
B_1=0\ ;\ B_2=\frac{i\delta n_1^2\mu^2\left(\alpha_1\beta_1-1\right)}{\alpha_2\beta_1\epsilon\left(\delta n_1^2\mu^2+2\right)}\ ,\ i=\sqrt{-1}\ ;\ \lambda=\pm\sqrt{-\frac{-2\alpha_1\beta_1-\delta n_1^2\mu^2}{\beta_1\left(2+\delta n_1^2\mu^2\right)}}\ ;\ \mu=\mu.\label{eq-46}
\end{align}
Substituting $A_0$ and $A_2$ from Eq. \eqref{eq-45}, $B_2$ and $\lambda$ from Eq. \eqref{eq-46} in the Eqs. \eqref{eq-26} and \eqref{eq-27} pertains the compound topological-non-topological soliton of Eq. \eqref{eq-1}.
\begin{align}
\Phi\left(x, t\right)=\frac{\delta n_1^2\mu^2}{\alpha_2\beta_1\epsilon}
\left\{
\frac{\left(\alpha_1\beta_1-1\right)
\left[\tanh^2\left(\xi_2\right)+i\tanh\left(\xi_2\right)\mbox{sech}\left(\xi_2\right)-\frac{2}{3}\right]}
{\left(2+\delta n_1^2\mu^2\right)}
\right\}\ ;\ i=\sqrt{-1}.\label{eq-47}
\end{align}
and compound singular soliton of Eq. \eqref{eq-1}. The 2D plots for the real part of the soliton is given in the Figure \ref{fig-2} with positive and negative value in $\xi_2$.
\begin{align}
\Phi\left(x, t\right):=\frac{\delta n_1^2\mu^2}{3\alpha_2\beta_1\epsilon}
\left\{
\frac{\left(\alpha_1\beta_1-1\right)
\left[3\coth^2\left(\xi_2\right)-3\coth\left(\xi_2\right)\mbox{csch}\left(\xi_2\right)-2\right]}
{\left(2+\delta n_1^2\mu^2\right)}
\right\}.\label{eq-48}
\end{align}
In the Eqs. \eqref{eq-47} and \eqref{eq-48} $\xi_2=\left(\mu\left[x\mp\sqrt{\frac{2\alpha_1\beta_1+\delta n_1^2\mu^2}{\beta_1\left(2+\delta n_1^2\mu^2\right)}}t\right]\right)$.
\end{case}

\begin{case}
\label{sol6}
When
\begin{align}
A_0=-\frac{2\delta n_1^2\mu^2\left(\alpha_1\beta_1-1\right)}{3\alpha_2\beta_1\epsilon\left(\delta n_1^2\mu^2+2\right)}\ ;\ A_1=0\ ;\ A_2=\frac{\delta n_1^2\mu^2\left(\alpha_1\beta_1-1\right)}{\alpha_2\beta_1\epsilon\left(\delta n_1^2\mu^2+2\right)}.\label{eq-49}
\end{align}
\begin{align}
B_1=0\ ;\ B_2=-\frac{i\delta n_1^2\mu^2\left(\alpha_1\beta_1-1\right)}{\alpha_2\beta_1\epsilon\left(\delta n_1^2\mu^2+2\right)}\ ,\ i=\sqrt{-1}\ ;\ \lambda=\pm\sqrt{-\frac{-2\alpha_1\beta_1-\delta n_1^2\mu^2}{\beta_1\left(2+\delta n_1^2\mu^2\right)}}\ ;\ \mu=\mu.\label{eq-50}
\end{align}
Now substituting $A_0$ and $A_2$ from the Eq. \eqref{eq-49}, $B_2$ and $\lambda$ from the Eq. \eqref{eq-50} into the Eqs. \eqref{eq-26} and \eqref{eq-27} yields the compound topological-non-topological soliton of Eq. \eqref{eq-1}. The 3D plots for the real and imaginary part of the soliton is given in the Figure \ref{fig-4}.
\begin{align}
\Phi\left(x, t\right)=-\frac{\delta n_1^2\mu^2}{\alpha_2\beta_1\epsilon}
\left\{
\frac{\left(\alpha_1\beta_1-1\right)
\left[-\tanh^2\left(\xi_2\right)+i\tanh\left(\xi_2\right)\mbox{sech}\left(\xi_2\right)+\frac{2}{3}\right]}
{\left(2+\delta n_1^2\mu^2\right)}
\right\}\ ;\ i=\sqrt{-1}.\label{eq-51}
\end{align}
and compound singular soliton of Eq. \eqref{eq-1}.
\begin{align}
\Phi\left(x, t\right)=\frac{\delta n_1^2\mu^2}{3\alpha_2\beta_1\epsilon}
\left\{
\frac{\left(\alpha_1\beta_1-1\right)
\left[3\coth^2\left(\xi_2\right)+3\coth\left(\xi_2\right)\mbox{csch}\left(\xi_2\right)-2\right]}
{\left(2+\delta n_1^2\mu^2\right)}
\right\}.\label{eq-52}
\end{align}
In the Eqs. \eqref{eq-51} and \eqref{eq-52} $\xi_2$ is given in the Eq. \eqref{eq-47}.
\end{case}

\section{Solitons obtained using modified exponential function method}
\label{soli2}

\begin{figure}[ht]
\centering
\par
\mbox{\subfigure[$\mathscr{R}\left(\Phi\left(x, t\right)\right)$ in Eq. \eqref{eq-56}]{\includegraphics[height=3in,width=3in]{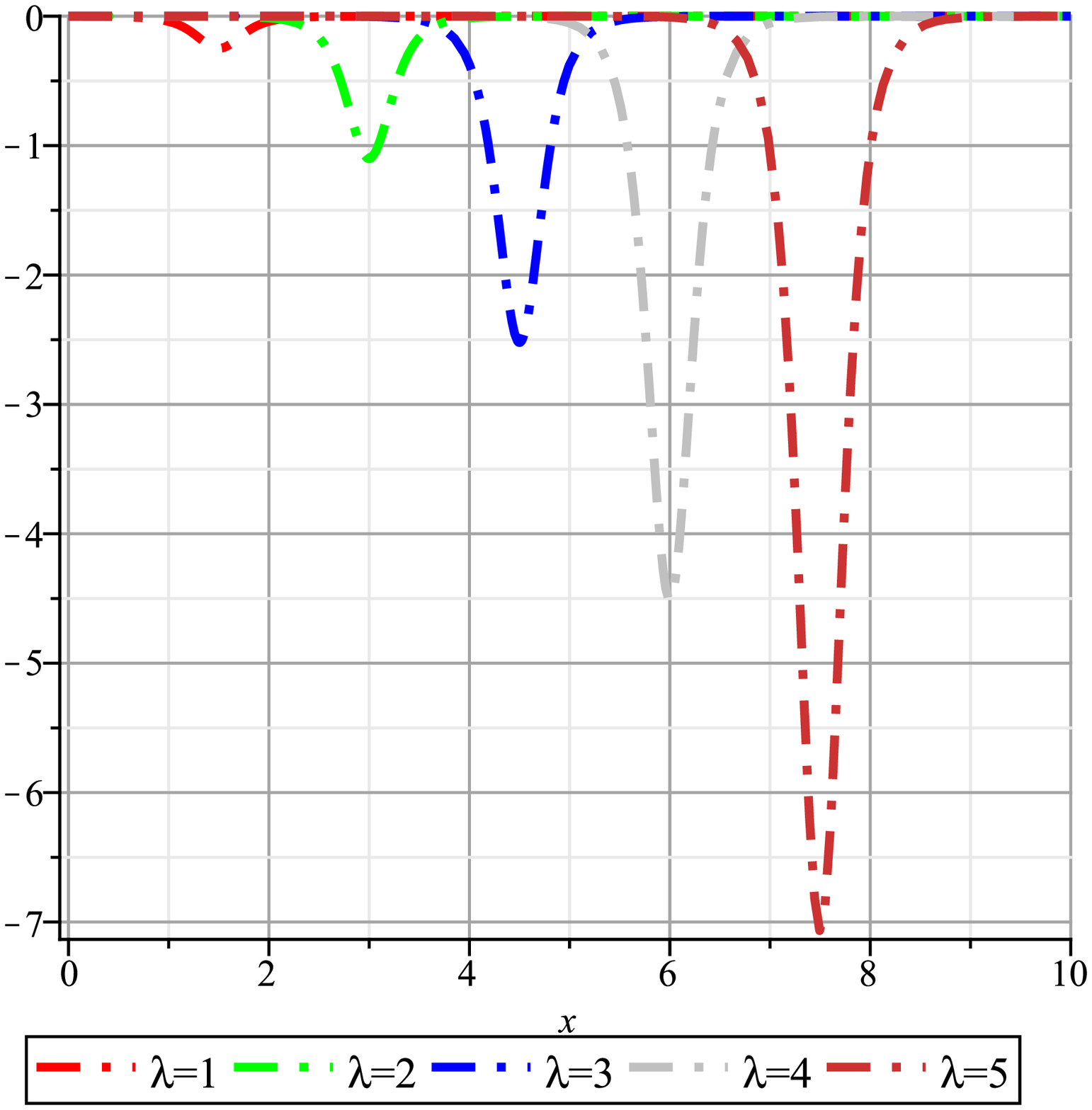}}\quad

\subfigure[$\mathscr{I}\left(\Phi\left(x, t\right)\right)$ in Eq. \eqref{eq-56}]{\includegraphics[height=3in,width=3in]{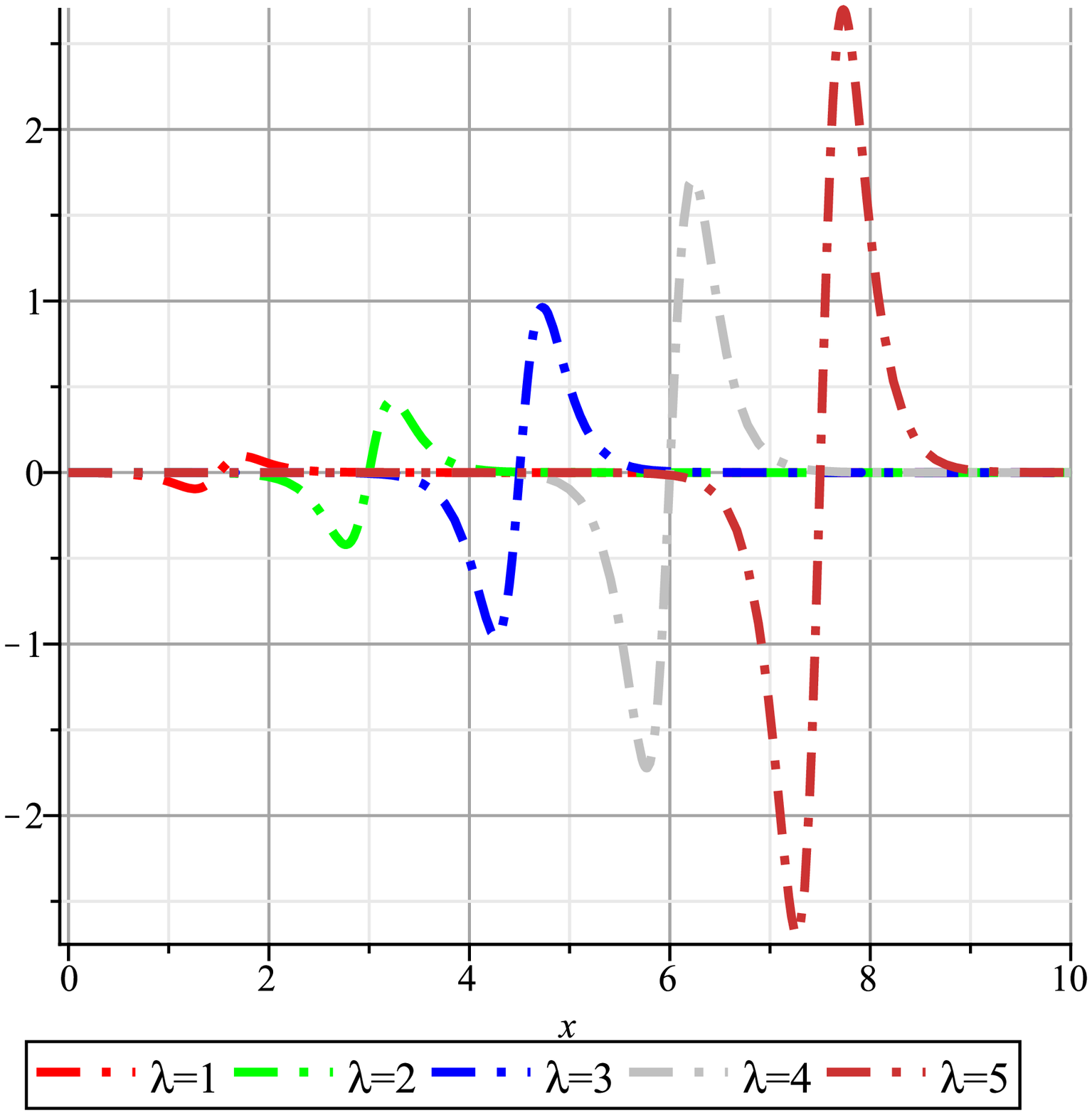} }}
\caption{\textbf{2D graph of the real and imaginary part of the soliton $\Phi\left(x, t\right)$ given in the Eq. \eqref{eq-56} in the domain $x\in[0\ ,\ 10]$.}}
\label{fig-5}
\end{figure}

\begin{figure}[ht]
\centering
\par
\mbox{\subfigure[$\mathscr{I}\left(\Phi\left(x, t\right)\right)$ in Eq. \eqref{eq-64}]{\includegraphics[height=3in,width=3in]{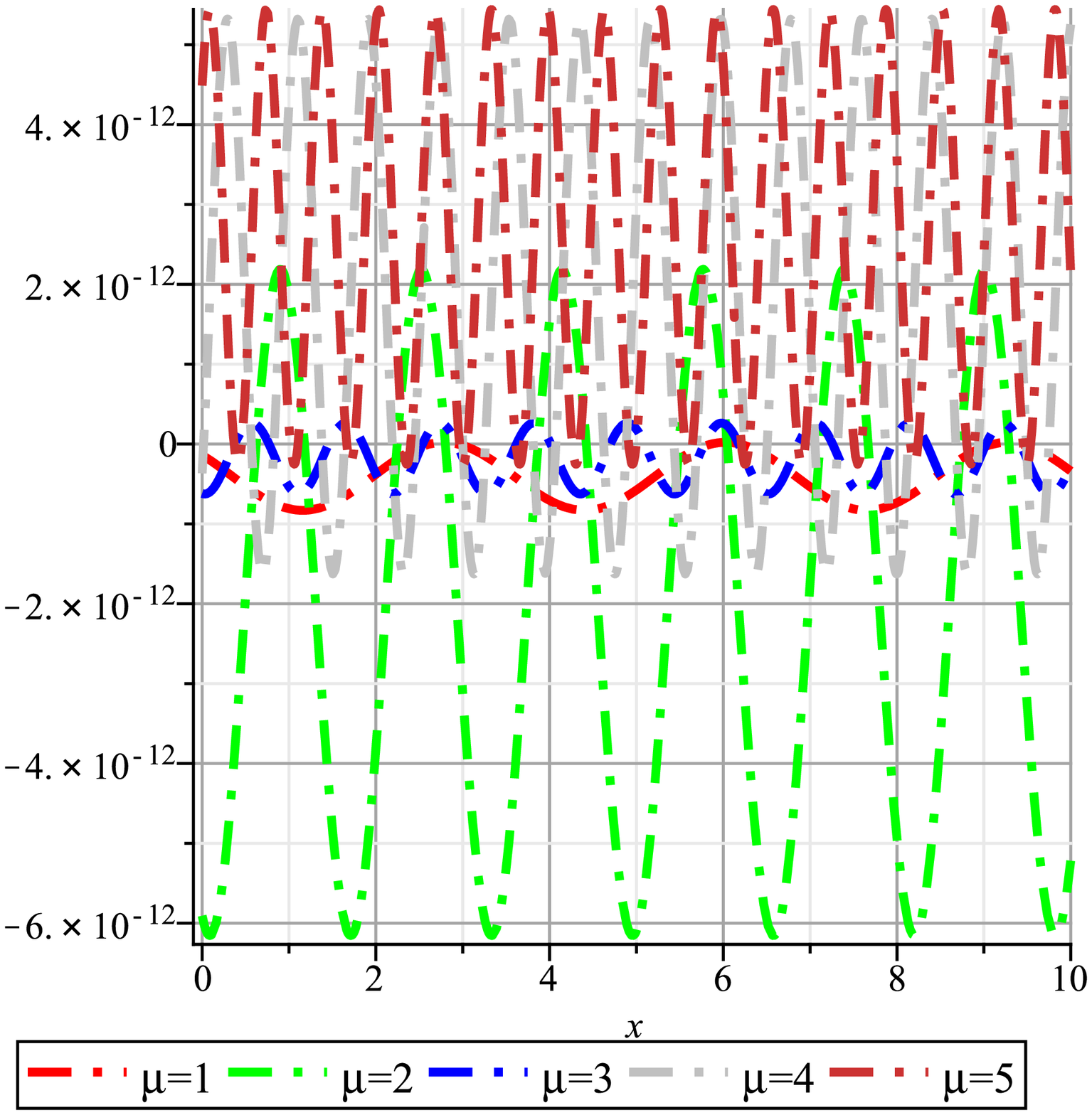}}\quad

\subfigure[$\mathscr{I}\left(\Phi\left(x, t\right)\right)$ in Eq. \eqref{eq-68}]{\includegraphics[height=3in,width=3in]{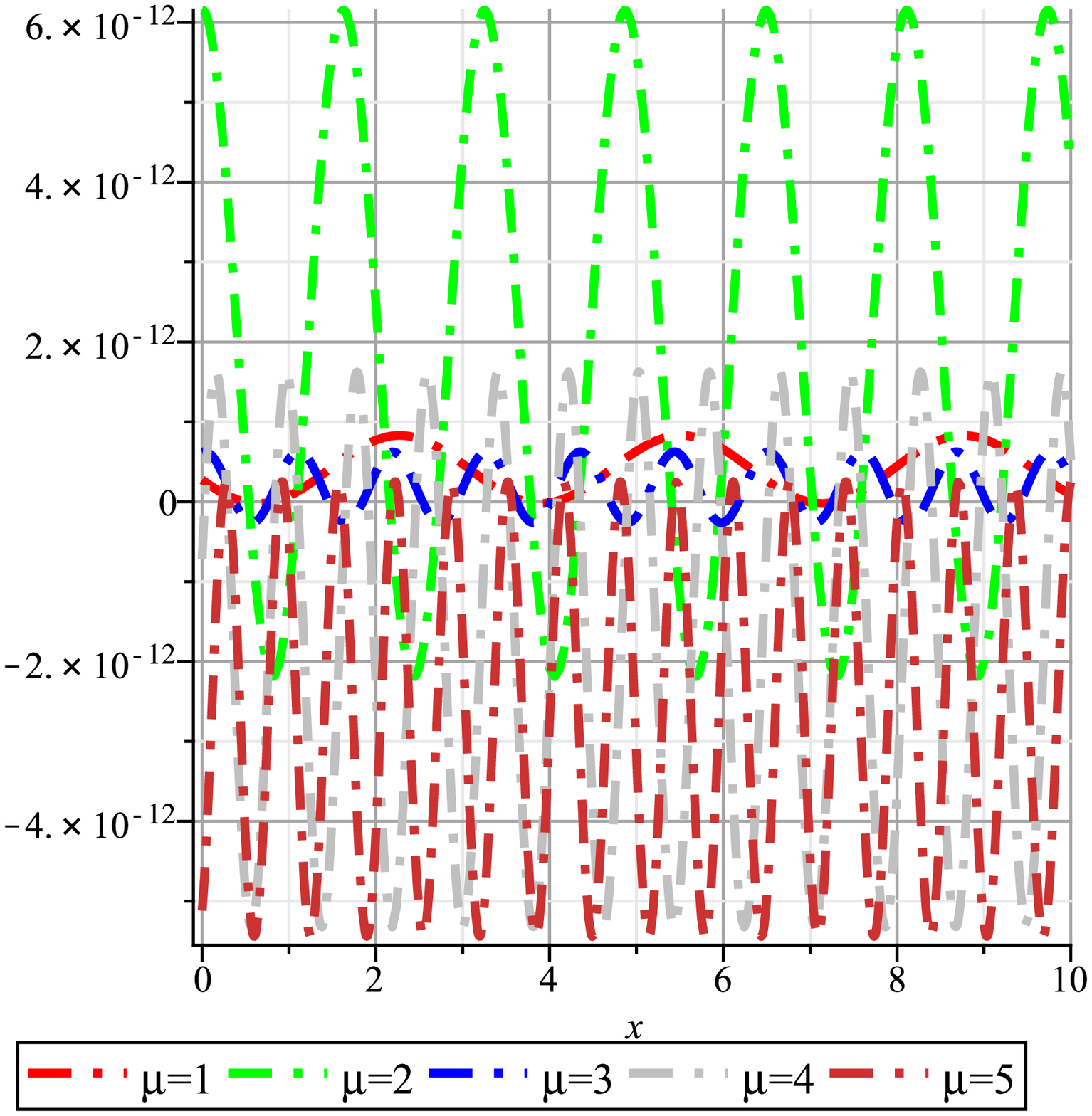} }}
\caption{\textbf{2D graph of the imaginary part of the soliton $\Phi\left(x, t\right)$ given in the Eqs. \eqref{eq-64} and \eqref{eq-68} in the domain $x\in[0\ ,\ 10]$.}}
\label{fig-6}
\end{figure}

\begin{figure}[ht]
\centering
\par
\mbox{\subfigure[$\mathscr{R}\left(\Phi\left(x, t\right)\right)$ in Eq. \eqref{eq-57}]{\includegraphics[height=3in,width=3in]{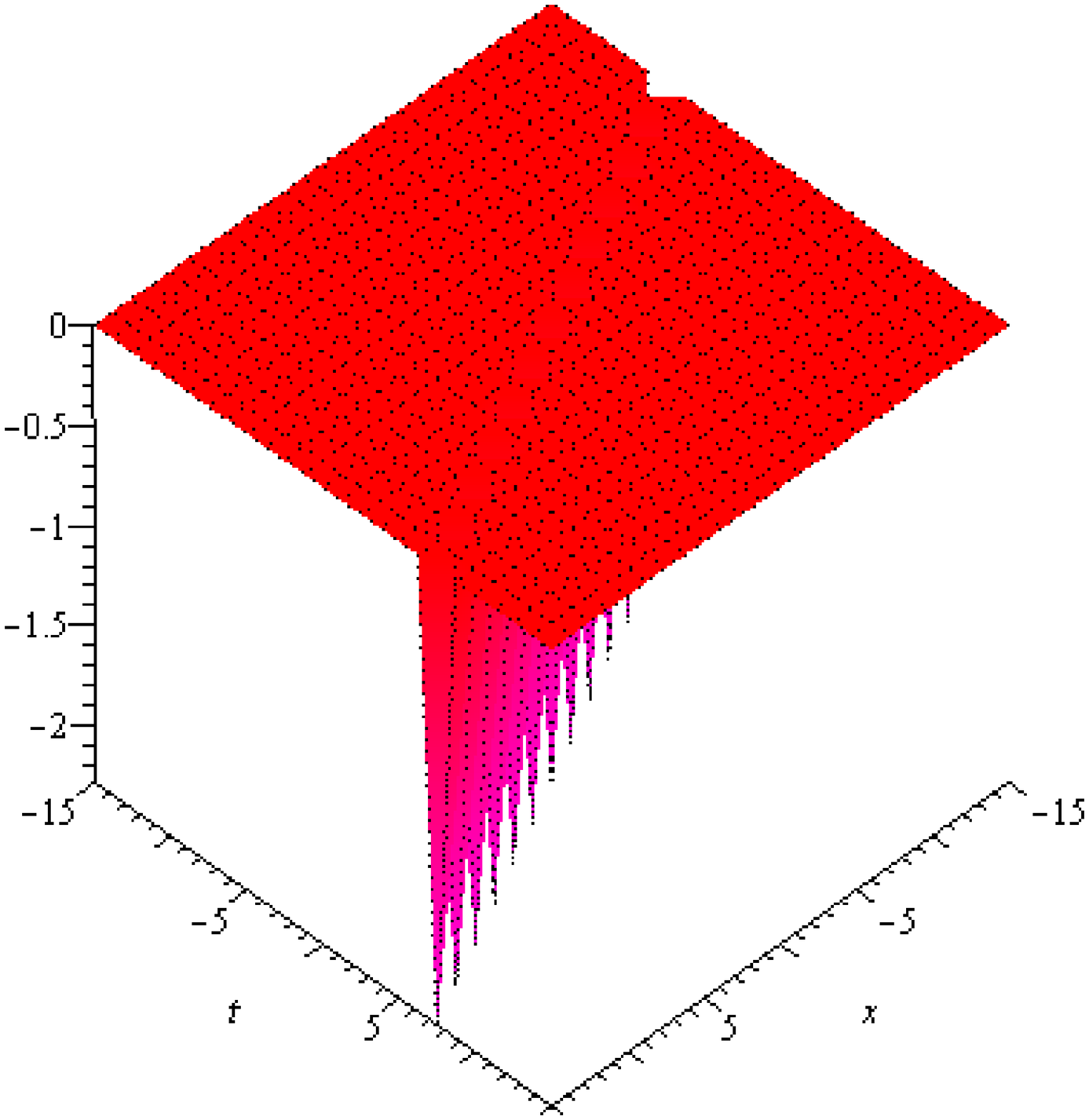}}\quad

\subfigure[$\mathscr{I}\left(\Phi\left(x, t\right)\right)$ in Eq. \eqref{eq-57}]{\includegraphics[height=3in,width=3in]{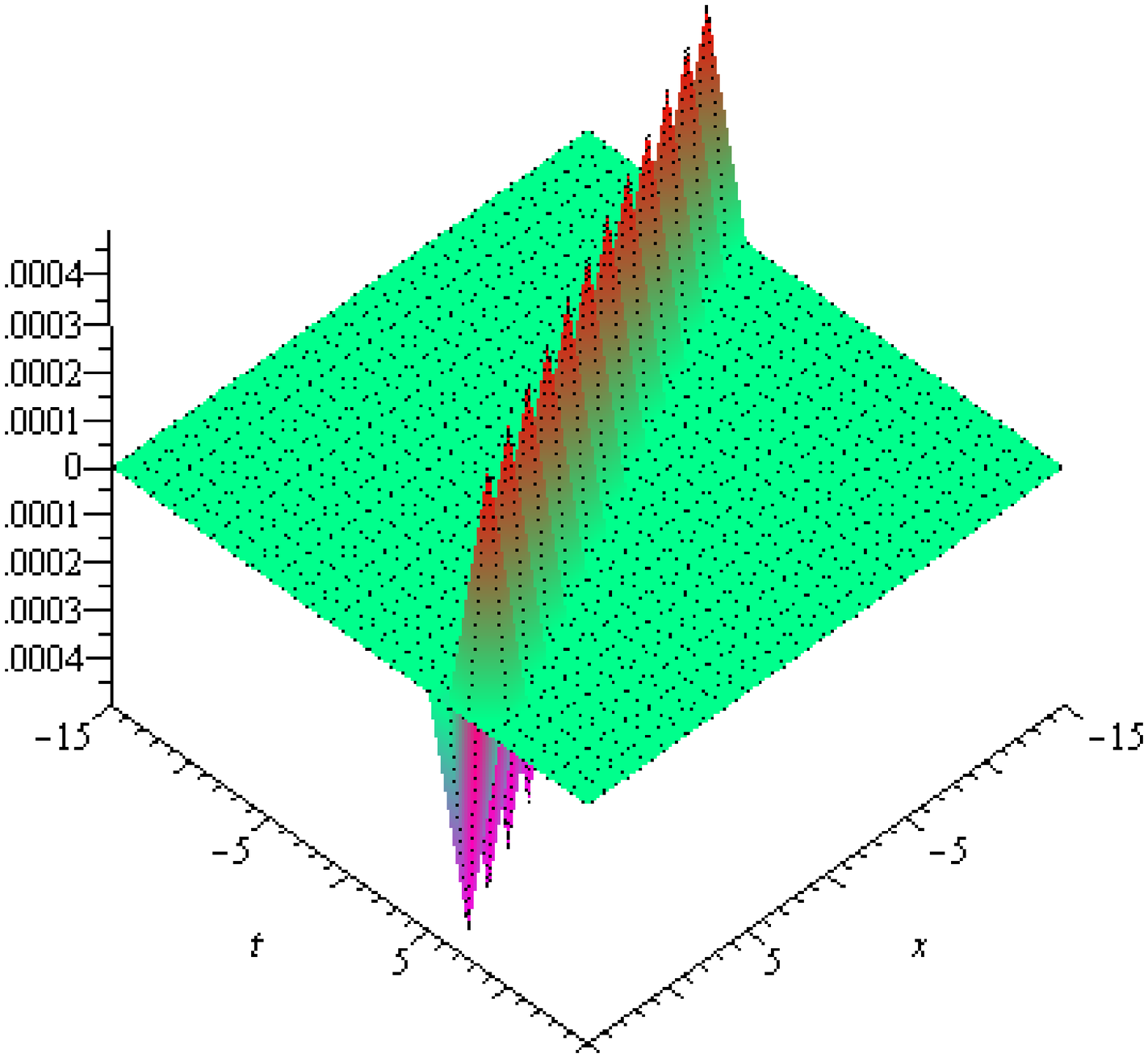} }}
\caption{\textbf{3D graph of the real and imaginary part of the soliton $\Phi\left(x, t\right)$ given in the Eq. \eqref{eq-57} in the domain $x\ ,\ t\ ,\ \in[-15\ ,\ 15]$.}}
\label{fig-7}
\end{figure}

\begin{figure}[ht]
\centering
\par
\mbox{\subfigure[$\mathscr{R}\left(\Phi\left(x, t\right)\right)$ in Eq. \eqref{eq-64}\ \mbox{with negative}\ $\xi$]{\includegraphics[height=3in,width=3in]{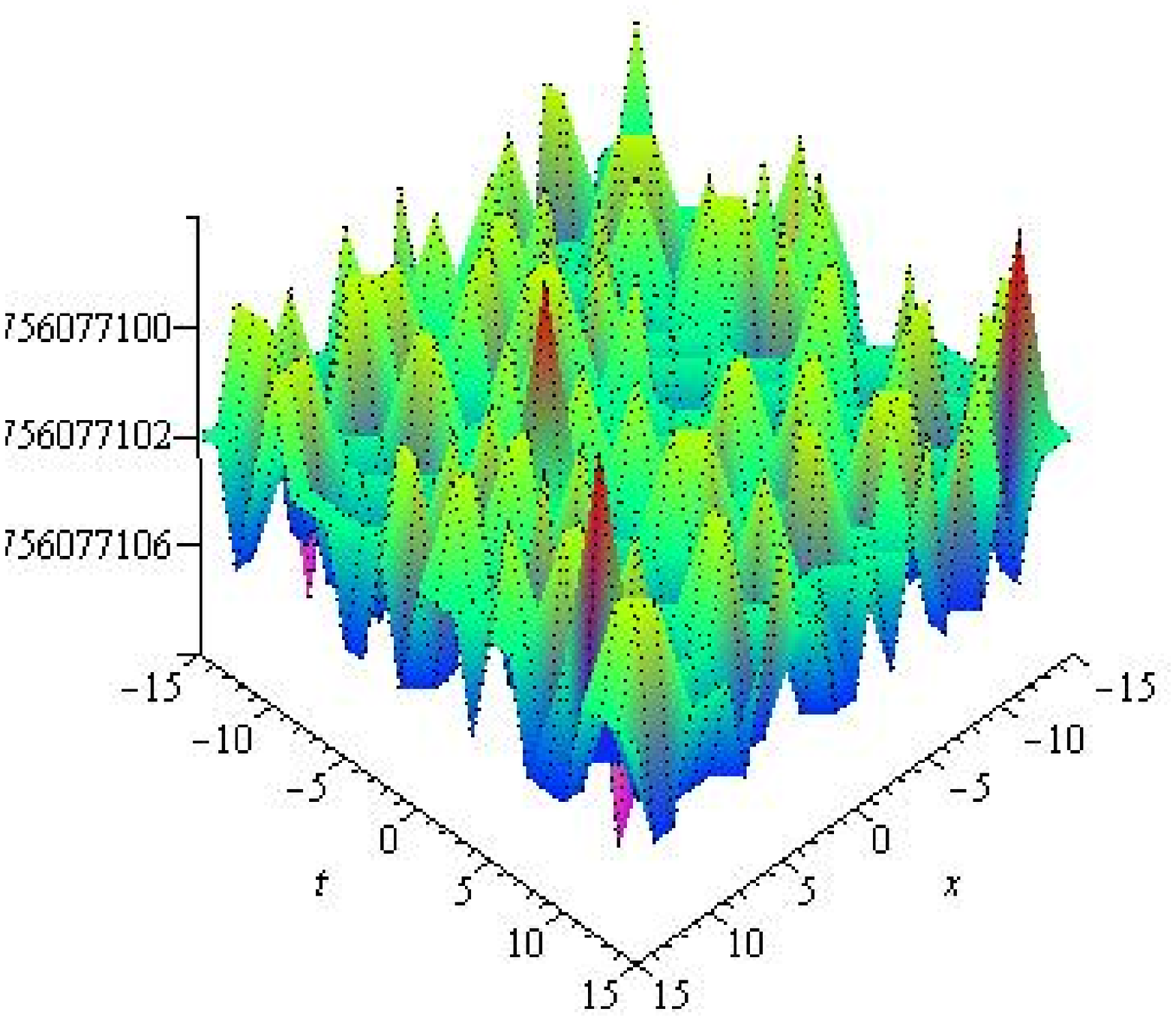}}\quad

\subfigure[$\mathscr{I}\left(\Phi\left(x, t\right)\right)$ in Eq. \eqref{eq-64}\ \mbox{with negative}\ $\xi$]{\includegraphics[height=3in,width=3in]{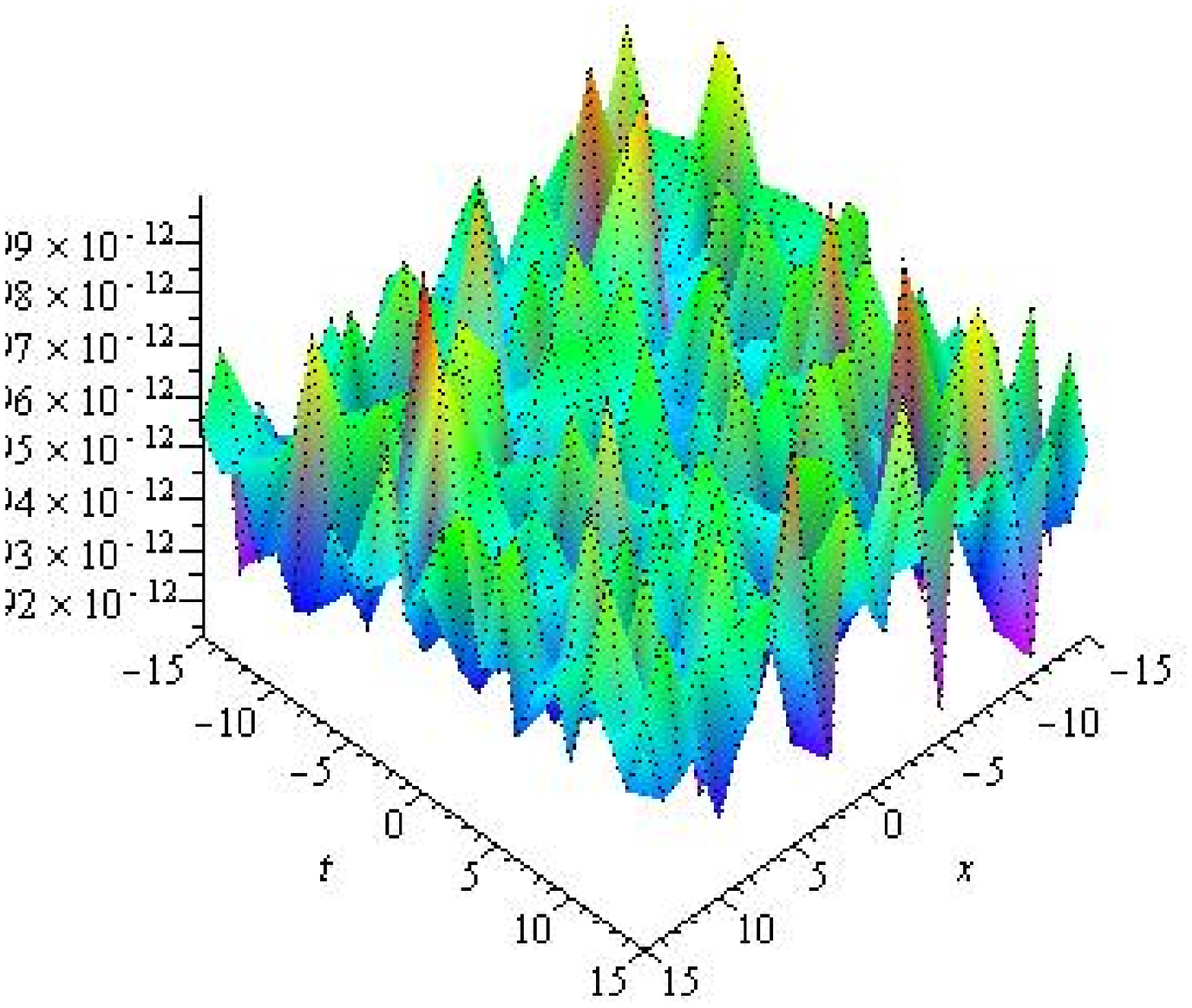} }}
\caption{\textbf{3D graph of the real and imaginary part of the soliton with negative $\xi_7$ in the $\Phi\left(x, t\right)$ given in the Eq. \eqref{eq-64} in the domain $x\ ,\ t\ ,\ \in[-15\ ,\ 15]$.}}
\label{fig-8}
\end{figure}

\begin{figure}[ht]
\centering
\par
\mbox{\subfigure[$\mathscr{R}\left(\Phi\left(x, t\right)\right)$ in Eq. \eqref{eq-64}\ \mbox{with positive}\ $\xi$]{\includegraphics[height=3in,width=3in]{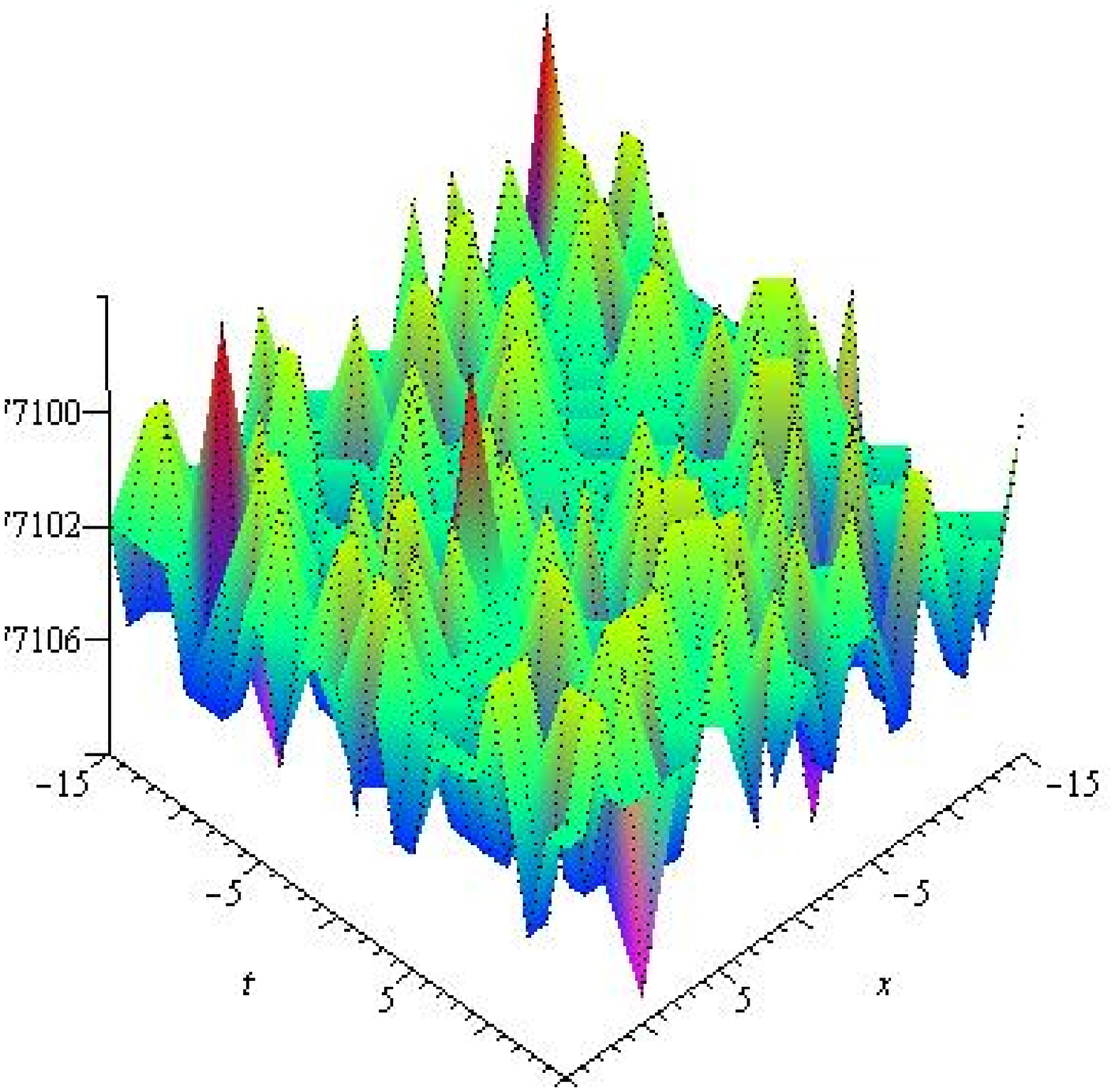}}\quad

\subfigure[$\mathscr{I}\left(\Phi\left(x, t\right)\right)$ in Eq. \eqref{eq-64}\ \mbox{with positive}\ $\xi$]{\includegraphics[height=3in,width=3in]{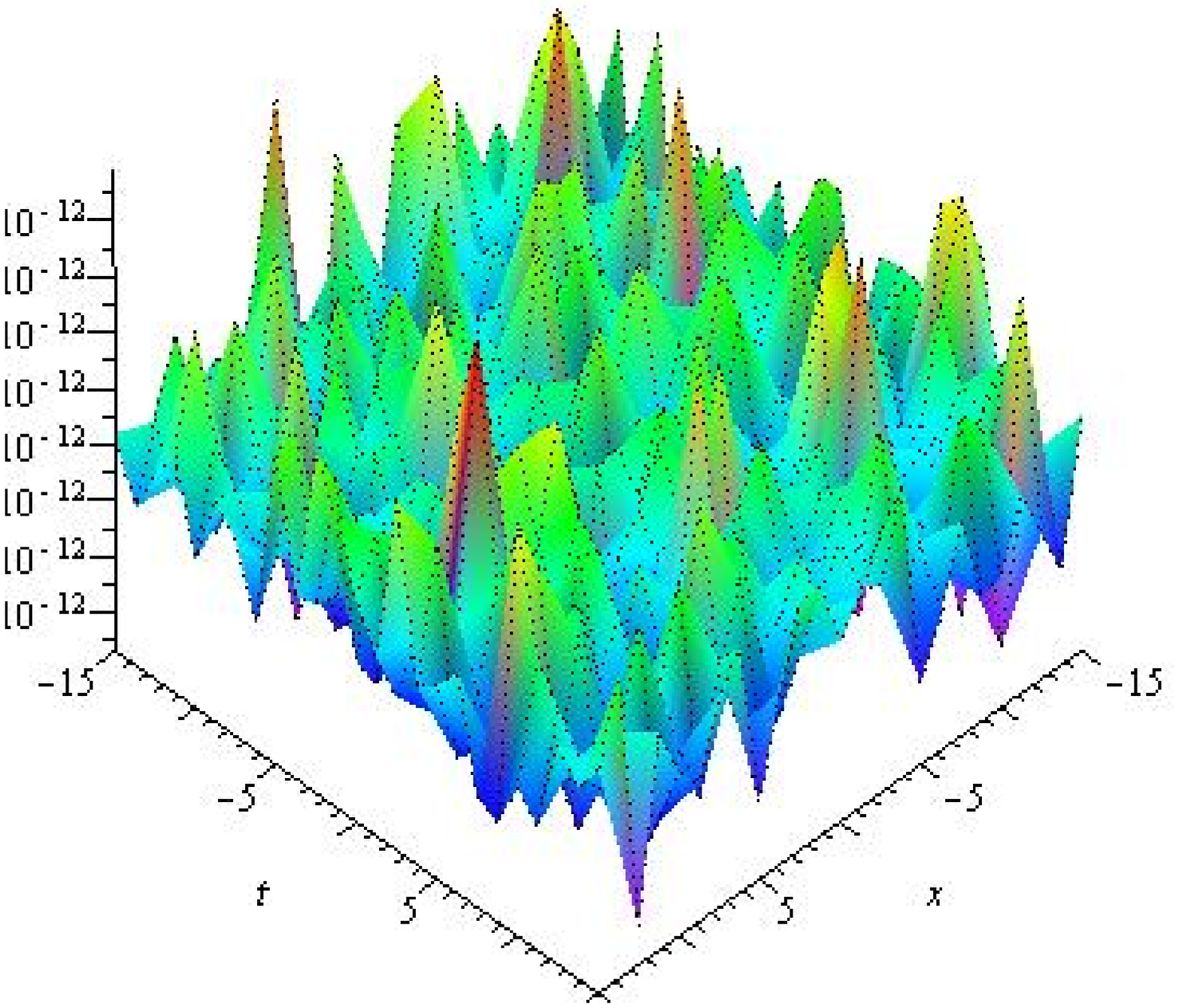} }}
\caption{\textbf{3D graph of the real and imaginary part of the soliton with positive $\xi_7$ in the $\Phi\left(x, t\right)$ given in the Eq. \eqref{eq-64} in the domain $x\ ,\ t\ ,\ \in[-15\ ,\ 15]$.}}
\label{fig-9}
\end{figure}

\begin{figure}[ht]
\centering
\par
\mbox{\subfigure[$\mathscr{R}\left(\Phi\left(x, t\right)\right)$ in Eq. \eqref{eq-68}\ \mbox{with negative}\ $\xi$]{\includegraphics[height=3in,width=3in]{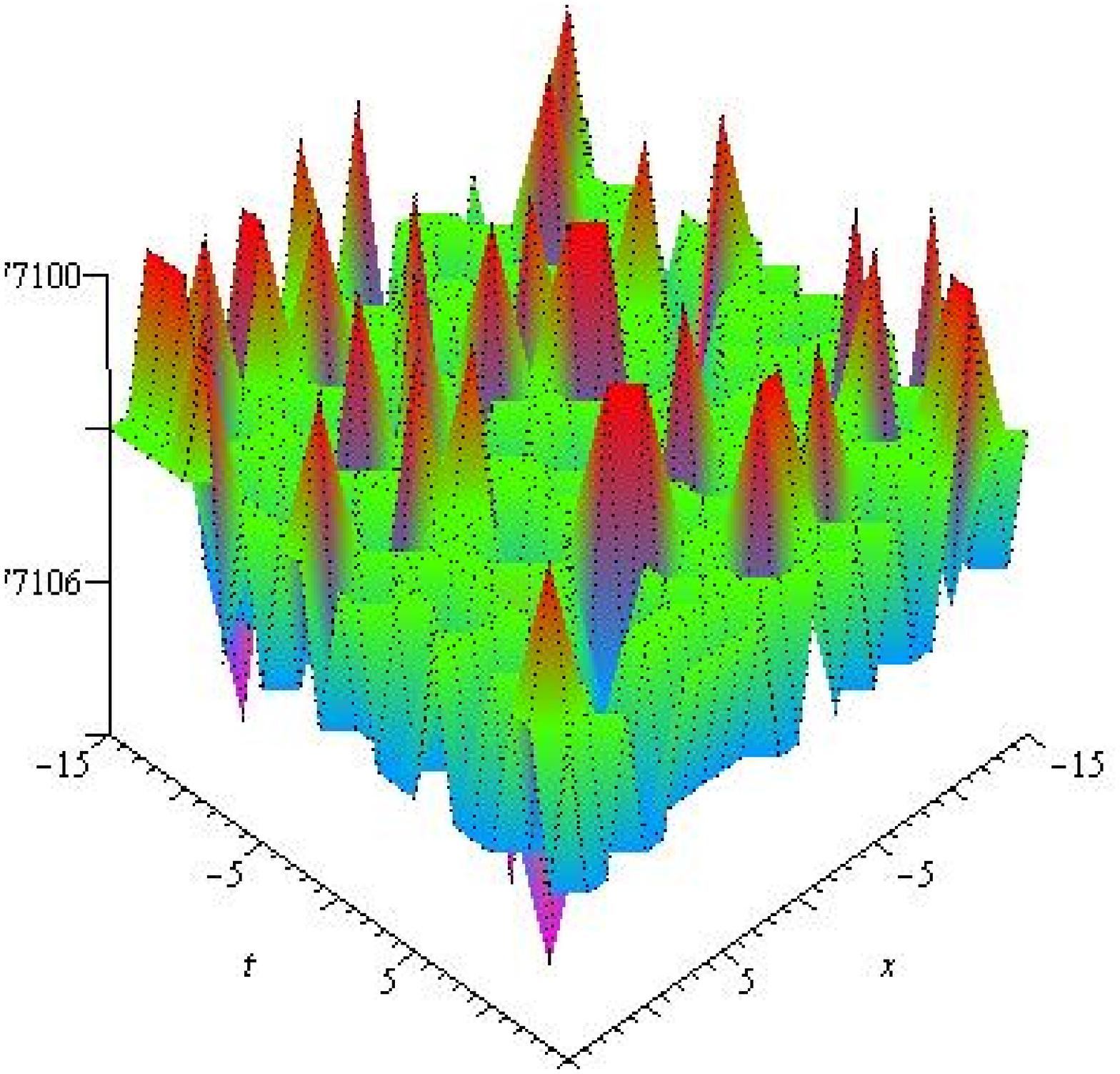}}\quad

\subfigure[$\mathscr{I}\left(\Phi\left(x, t\right)\right)$ in Eq. \eqref{eq-68}\ \mbox{with negative}\ $\xi$]{\includegraphics[height=3in,width=3in]{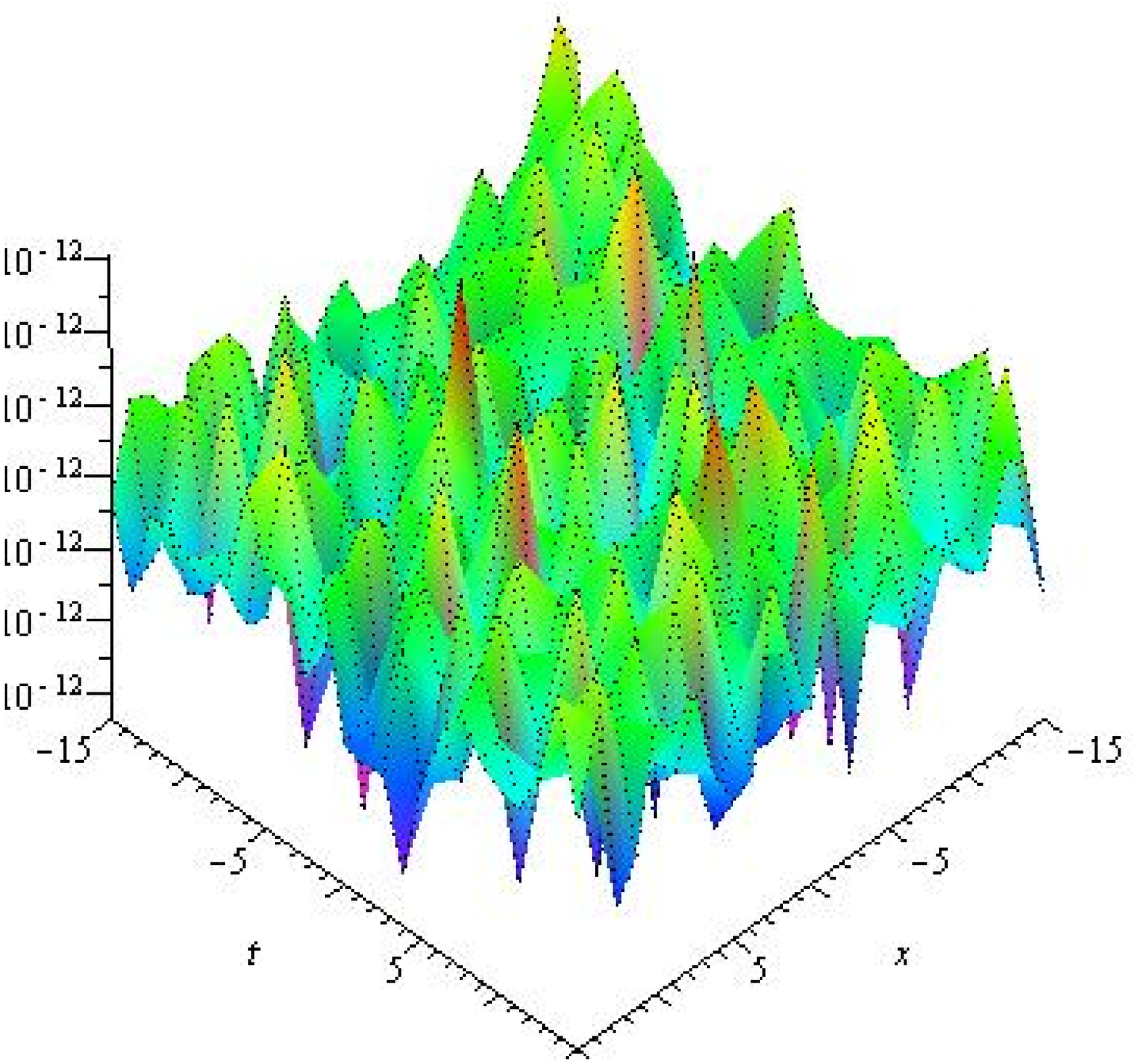} }}
\caption{\textbf{3D graph of the real and imaginary part of the soliton with negative $\xi_7$ in the $\Phi\left(x, t\right)$ given in the Eq. \eqref{eq-68} in the domain $x\ ,\ t\ ,\ \in[-15\ ,\ 15]$.}}
\label{fig-10}
\end{figure}

\begin{figure}[ht]
\centering
\par
\mbox{\subfigure[$\mathscr{R}\left(\Phi\left(x, t\right)\right)$ in Eq. \eqref{eq-68}\ \mbox{with positive}\ $\xi$]{\includegraphics[height=3in,width=3in]{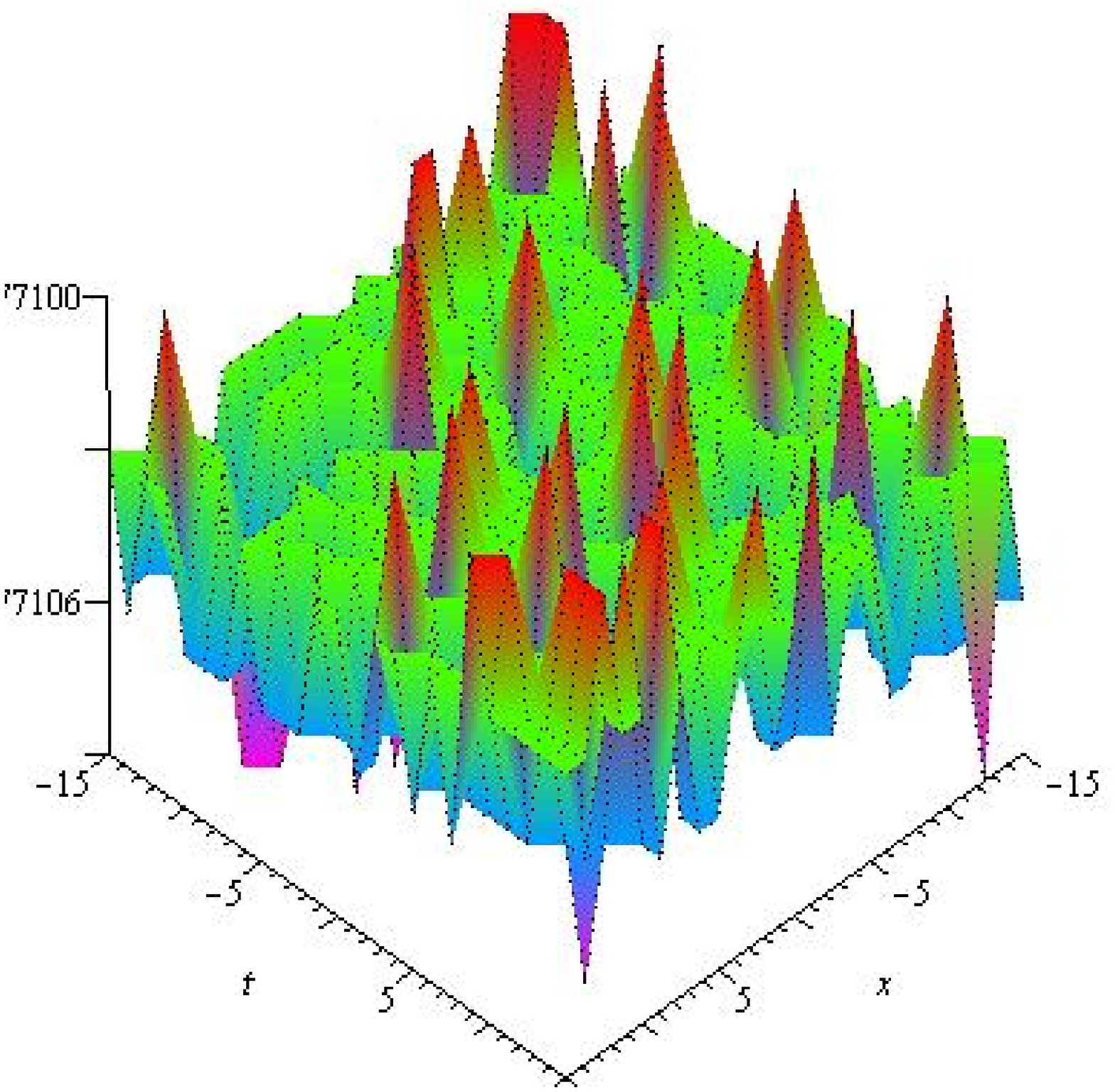}}\quad

\subfigure[$\mathscr{I}\left(\Phi\left(x, t\right)\right)$ in Eq. \eqref{eq-68}\ \mbox{with positive}\ $\xi$]{\includegraphics[height=3in,width=3in]{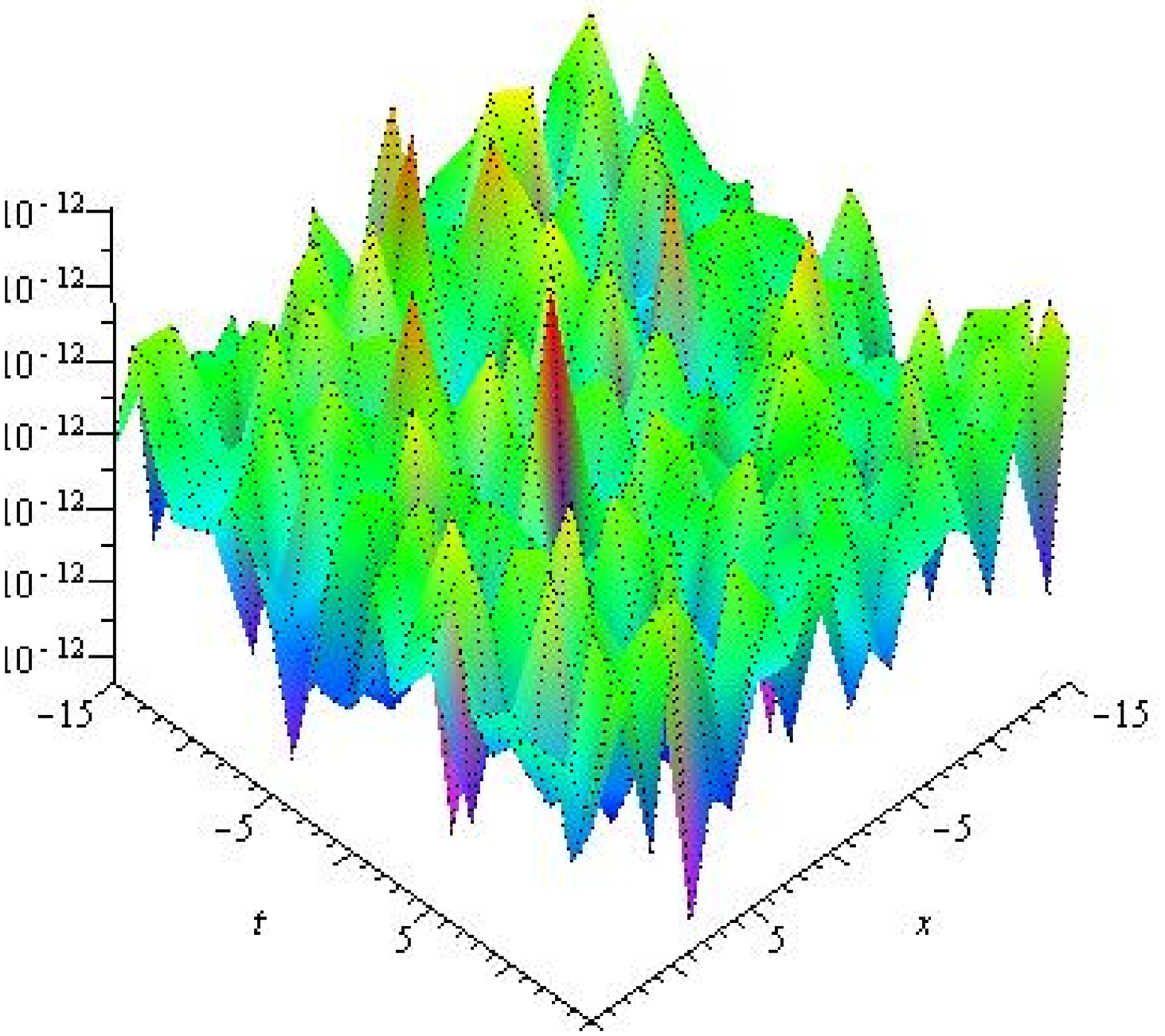} }}
\caption{\textbf{3D graph of the real and imaginary part of the soliton with positive $\xi_7$ in the $\Phi\left(x, t\right)$ given in the Eq. \eqref{eq-68} in the domain $x\ ,\ t\ ,\ \in[-15\ ,\ 15]$.}}
\label{fig-11}
\end{figure}

The balancing of $\frac{d^2u\left(\xi\right)}{d\xi^2}$ and $u^2\left(\xi\right)$ in the Eq. \eqref{eq-25} gives $N=M+2$, taking $M=1$ gives $N=3$. Therefore the initial solution of Eq. \eqref{eq-25} is assumed as from the Eq. \eqref{eq-18}.
\begin{align}
u\left(\xi\right)=\frac{P_0+P_1\exp\left(-\varphi\left(\xi\right)\right)
+P_2\exp\left(-2\varphi\left(\xi\right)\right)+P_3\exp\left(-3\varphi\left(\xi\right)\right)}
{Q_0+Q_1\exp\left(-\varphi\left(\xi\right)\right)}.\label{eq-53}
\end{align}
Now substitutuing the Eq. \eqref{eq-53} along with Eq. \eqref{eq-19} in the Eq. \eqref{eq-25} results in the polynomial $\exp\left(-i\varphi\left(\xi\right)\right)\ ;\ i=0\ ,\ 1,\ , \cdots\ ,\ 7$. Next collecting the coefficent of each $i$ gives systems of eight equations. Solving the aforesaid overdetermined equations gives the unknowns $P_i\ ;\ i=0 - 3$, $Q_j\ ;\ j=0\ ,\ 1$, $\lambda$ and $\mu$. Substituting the obtained unknowns in the Eq. \eqref{eq-53} yields the solitons of the Eq. \eqref{eq-1}. Each solitons are given in the following cases.

\begin{case}
\label{sol7}
When
\begin{align}
P_0=\frac{2Q_0\sigma\left(\lambda^2-\alpha_1\right)}
{\alpha_2\epsilon\left(\tau^2-4\sigma\right)}\ ;\ P_1=\frac{2\left(\lambda^2-\alpha_1\right)\left(Q_0\tau+Q_1\sigma\right)}
{\alpha_2\epsilon\left(\tau^2-4\sigma\right)}\ ;\ P_2=\frac{2\left(\lambda^2-\alpha_1\right)\left(Q_0+Q_1\tau\right)}
{\alpha_2\epsilon\left(\tau^2-4\sigma\right)}.\label{eq-54}
\end{align}
\begin{align}
P_3=\frac{2Q_1\left(\lambda^2-\alpha_1\right)}
{\alpha_2\epsilon\left(\tau^2-4\sigma\right)}\ ;\ Q_0=Q_0\ ;\ Q_1=Q_1\ ;\ \lambda=\lambda\ ;\ \mu=\pm\frac{\sqrt{2}\sqrt{\beta_1}\sqrt{\lambda^2-\alpha_1}}
{n_1\sqrt{\delta}\sqrt{\lambda^2\beta_1-1}\sqrt{\tau^2-4\sigma}}.\label{eq-55}
\end{align}
Substituting $P_0\ ,\ P_1\ ,\ P_2$ from the Eq. \eqref{eq-54}, $P_3\ ,\ Q_0\ ,\ Q_1\ ,\ \lambda\ ,\ \mu$ from the Eq. \eqref{eq-55} into the Eq. \eqref{eq-53} gives the soliton-like solution when $\sigma\ne 0$ and $\tau^2-4\sigma>0$. The 2D plots for the real and imaginary part of the soliton is given in the Figure \ref{fig-5}.
\begin{align}
\Phi\left(x, t\right)=\frac{2\sigma\left(\lambda^2-\alpha_1\right)}{\alpha_2\epsilon\left(\tau^2-4\sigma\right)}
\left\{
\frac{\left[Q_0\left(\tau^2-4\sigma\right)^{3/2}\tanh\left(\xi_3\right)
+\left(\tau^2-4\sigma\right)\left(\tau Q_0-2\sigma Q_1\right)\right]\left(\tanh^2\left(\xi_3\right)-1\right)}
{\left(\sqrt{\tau^2-4\sigma}\tanh\left(\xi_3\right)+\tau\right)^2
\left(Q_0\sqrt{\tau^2-4\sigma}\tanh\left(\xi_3\right)+\tau Q_0-2\sigma Q_1\right)}
\right\}.\label{eq-56}
\end{align}
In the Eq. \eqref{eq-56} $\xi_3=\left\{\frac{\sqrt{\tau^2-4\sigma}}{2}\left[\pm\frac{\sqrt{2}\sqrt{\beta_1}
\sqrt{\lambda^2-\alpha_1}\left(x-\lambda t\right)}{n_1\sqrt{\delta}\sqrt{\lambda^2\beta_1-1}\sqrt{\tau^2-4\sigma}}+e\right]\right\}$.
And singular periodic wave solution when $\sigma\ne 0$ and $\tau^2-4\sigma<0$. The 3D plot for the real and imaginary part of the soliton is given in the Figure \ref{fig-7}.
\begin{align}
\Phi\left(x, t\right)=&-\left\{\frac{2\sigma\left(\lambda^2-\alpha_1\right)}
{\alpha_2\epsilon\left(\tau^2-4\sigma\right)}\right\}\nonumber\\
&\times\left\{
\frac{\left[-Q_0\left(-\tau^2+4\sigma\right)^{3/2}\tan^3\left(\xi_4\right)
+\left(Q_0\sqrt{-\tau^2+4\sigma}\tan\left(\xi_4\right)+\left(\tan^2\left(\xi_4\right)+1\right)
\left(\tau Q_0-2\sigma Q_1\right)\right)\left(\tau^2-4\sigma\right)\right]}
{\left(\sqrt{-\tau^2+4\sigma}\tan\left(\xi_4\right)+\tau\right)^2
\left(Q_0(\sqrt{-\tau^2+4\sigma}\tan\left(\xi_4\right)+\tau Q_0-2\sigma Q_1\right)}
\right\}.\label{eq-57}
\end{align}
In the Eq. \eqref{eq-57} $\xi_4=\left\{\frac{\sqrt{-\tau^2+4\sigma}}{2}\left[\pm\frac{\sqrt{2}\sqrt{\beta_1}
\sqrt{\lambda^2-\alpha_1}\left(x-\lambda t\right)}{n_1\sqrt{\delta}\sqrt{\lambda^2\beta_1-1}\sqrt{\tau^2-4\sigma}}+e\right]\right\}$.
\end{case}

\begin{case}
\label{sol8}
When
\begin{align}
P_0=-\frac{Q_0\left(\tau^2+2\sigma\right)\left(\lambda^2-\alpha_1\right)}
{3\alpha_2\epsilon\left(\tau^2-4\sigma\right)}\ ;\ P_1=-\frac{\left(\lambda^2-\alpha_1\right)\left(6Q_0\tau+\left(\tau^2+2\sigma\right)Q_1\right)}
{3\alpha_2\epsilon\left(\tau^2-4\sigma\right)}\ ;\ P_2=-\frac{2\left(\lambda^2-\alpha_1\right)\left(Q_0+Q_1\tau\right)}
{\alpha_2\epsilon\left(\tau^2-4\sigma\right)}.\label{eq-58}
\end{align}
\begin{align}
P_3=-\frac{2Q_1\left(\lambda^2-\alpha_1\right)}
{\alpha_2\epsilon\left(\tau^2-4\sigma\right)}\ ;\ Q_0=Q_0\ ;\ Q_1=Q_1\ ;\ \lambda=\lambda\ ;\ \mu=\pm\frac{i\sqrt{2}\sqrt{\beta_1}\sqrt{\lambda^2-\alpha_1}}
{n_1\sqrt{\delta}\sqrt{\lambda^2\beta_1-1}\sqrt{\tau^2-4\sigma}}\ ;\ i=\sqrt{-1}.\label{eq-59}
\end{align}
Now substituting $P_0\ ,\ P_1\ ,\ P_2$ from the Eq. \eqref{eq-58}, $P_3\ ,\ Q_0\ ,\ Q_1\ ,\ \lambda\ ,\ \mu$ from the Eq. \eqref{eq-59} in the Eq. \eqref{eq-53} pertains the soliton-like solution when $\sigma\ne 0$ and $\tau^2-4\sigma>0$.
\begin{align}
\Phi\left(x, t\right)=&-\left\{\frac{\left(\lambda^2-\alpha_1\right)}{\alpha_2\epsilon\left(\tau^2-4\sigma\right)
\left(\sqrt{\tau^2-4\sigma}\tanh\left(\xi_5\right)+\tau\right)^2
\left(Q_0\sqrt{\tau^2-4\sigma}\tanh\left(\xi_5\right)+Q_0\tau-2Q_1\sigma\right)}\right\}\nonumber\\
&\times\left\{
\frac{1}{3}\left(\tau^2-4\sigma\right)^{3/2}\left[Q_0\left(\tau^2+2\sigma\right)\tanh^2\left(\xi_5\right)
-6Q_0\sigma+3Q_0\tau^2-4Q_1\sigma\tau\right]\tanh\left(\xi_5\right)\right.\nonumber\\
&\left.+\left(\tau^2-4\sigma\right)\left[
\left(Q_0\tau^3-\frac{2}{3}Q_1\tau^2\sigma-2Q_0\sigma\tau-\frac{4}{3}Q_1\sigma^2\right)\tanh^2\left(\xi_5\right)
+\frac{1}{3}\left(\tau^2-6\sigma\right)\left(Q_0\tau-2Q_1\sigma\right)\right]
\right\}.\label{eq-60}
\end{align}
In the Eq. \eqref{eq-60} $\xi_5=\left\{\frac{\sqrt{\tau^2-4\sigma}}{2}\left[\pm\frac{i\sqrt{2}\sqrt{\beta_1}
\sqrt{\lambda^2-\alpha_1}\left(x-\lambda t\right)}{n_1\sqrt{\delta}\sqrt{\lambda^2\beta_1-1}\sqrt{\tau^2-4\sigma}}+e\right]\right\}\ ;\ i=\sqrt{-1}$.
And singular periodic solution when $\sigma\ne 0$ and $\tau^2-4\sigma<0$.
\begin{align}
\Phi\left(x, t\right)=&-\left\{\frac{\left(\lambda^2-\alpha_1\right)}{\alpha_2\epsilon\left(\tau^2-4\sigma\right)
\left(\sqrt{-\tau^2+4\sigma}\tan\left(\xi_6\right)+\tau\right)^2
\left(Q_0\sqrt{-\tau^2+4\sigma}\tan\left(\xi_6\right)+Q_0\tau-2Q_1\sigma\right)}\right\}\nonumber\\
&\times\left\{
-\frac{1}{3}\left(-\tau^2+4\sigma\right)^{3/2}\left(\tau^2+2\sigma\right)\tan^3\left(\xi_6\right)\right.\nonumber\\
&\left.+\left(\tau^2-4\sigma\right)
\left[
\left(Q_0\tau^2-\frac{4}{3}Q_1\tau\sigma-2Q_0\sigma\right)\sqrt{-\tau^2+4\sigma}\tan\left(\xi_6\right)\right]\right.\nonumber\\
&\left.+\left(\tau^2-4\sigma\right)\left[\left(Q_0\tau^3-\frac{2}{3}Q_1\tau^2\sigma-2Q_0\tau\sigma-\frac{4}{3}Q_1\sigma^2\right)\tan^2\left(\xi_6\right)
-\frac{1}{3}\left(\tau^2-6\sigma\right)\left(Q_0\tau-2Q_1\sigma\right)
\right]
\right\}.\label{eq-61}
\end{align}
In the Eq. \eqref{eq-61} $\xi_6=\left\{\frac{\sqrt{-\tau^2+4\sigma}}{2}\left[\pm\frac{i\sqrt{2}\sqrt{\beta_1}
\sqrt{\lambda^2-\alpha_1}\left(x-\lambda t\right)}{n_1\sqrt{\delta}\sqrt{\lambda^2\beta_1-1}\sqrt{\tau^2-4\sigma}}+e\right]\right\}\ ;\ i=\sqrt{-1}$.
\end{case}

\begin{case}
\label{sol9}
When
\begin{align}
P_0=\frac{n_1^2\mu^2\delta Q_1\sqrt{\tau^2-4\sigma}\left(\alpha_1\beta_1-1\right)\left(\tau^2+\sqrt{\tau^2-4\sigma}\tau-4\sigma\right)}
{6\alpha_2\beta_1\epsilon\left(2+n_1^2\mu^2\tau^2\delta-4n_1^2\mu^2\delta\sigma\right)}\ ;\ P_1=\frac{n_1^2\mu^2\delta Q_1\left(\alpha_1\beta_1-1\right)\left(\tau^2-4\sigma\right)}
{3\alpha_2\beta_1\epsilon\left(2+n_1^2\mu^2\tau^2\delta-4n_1^2\mu^2\delta\sigma\right)}.\label{eq-62}
\end{align}
\begin{align}
P_2=P_3=0\ ,\ Q_0=\frac{1}{2}\left(\tau+\sqrt{\tau^2-4\sigma}\right)Q_1\ ;\ Q_1=Q_1\ ;\ \lambda=\pm\frac{\sqrt{2\alpha_1\beta_1-4n_1^2\mu^2\delta\sigma+n_1^2\mu^2\tau^2\delta}}
{\sqrt{\beta_1}\sqrt{2-4n_1^2\mu^2\delta\sigma+n_1^2\mu^2\tau^2\delta}}\ ;\ \mu=\mu.\label{eq-63}
\end{align}
Now substituting $P_0\ ,\ P_1$ from the Eq. \eqref{eq-62}, $Q_0\ ,\ Q_1\ ,\ \lambda\ ,\ \mu$ from the Eq. \eqref{eq-63} in the Eq. \eqref{eq-53} gives the soliton-like solution when $\sigma\ne 0$ and $\tau^2-4\sigma>0$. The 2D plot for the imaginary part of the soliton is given in the Figure \ref{fig-6}. The 3D plot for the real and imaginary part of the soliton is given in the Figure \ref{fig-8} with the negative value of $\xi_7$ and with the positive value of $\xi_7$ in the Figure \ref{fig-9}.
\begin{align}
\Phi\left(x, t\right)=&\left\{\frac{n_1^2\mu^2\delta\left(\alpha_1\beta_1-1\right)}
{3\alpha_2\beta_1\epsilon\left(2+n_1^2\mu^2\tau^2\delta-4n_1^2\mu^2\delta\sigma\right)}\right\}\nonumber\\
&\times\left\{
\frac{\left(\tau^4-8\tau^2\sigma+\left(\tau^2-4\sigma\right)^{3/2}\tau+16\sigma^2\right)\tanh\left(\xi_7\right)
+\left(\tau^2-4\sigma\right)\left(\tau^2+\sqrt{\tau^2-4\sigma}\tau-4\sigma\right)}
{\left(\tanh\left(\xi_7\right)+1\right)\left(\tau^2+\sqrt{\tau^2-4\sigma}\tau-4\sigma\right)}
\right\}.\label{eq-64}
\end{align}
In the Eq. \eqref{eq-64} $\xi_7=\left\{\frac{\sqrt{\tau^2-4\sigma}}{2}\left[\mu\left(x\mp\frac{\sqrt{n_1^2\mu^2\tau^2\delta
-4n_1^2\mu^2\delta\sigma+2\alpha_1\beta_1}t}
{\sqrt{\beta_1}\sqrt{2+n_1^2\mu^2\tau^2\delta-4n_1^2\mu^2\delta\sigma}}\right)\right]\right\}$.
And rational function solution when $\sigma\ne 0$ and $\tau^2-4\sigma<0$.
\begin{align}
\Phi\left(x, t\right)=\left\{\frac{n_1^2\mu^2\delta\left(\alpha_1\beta_1-1\right)\left(\tau^2-4\sigma\right)}
{3\alpha_2\beta_1\epsilon\left(2+n_1^2\mu^2\tau^2\delta-4n_1^2\mu^2\delta\sigma\right)}\right\}.\label{eq-65}
\end{align}
\end{case}

\begin{case}
\label{sol10}
When
\begin{align}
P_0=\frac{n_1^2\mu^2\delta Q_1\sqrt{\tau^2-4\sigma}\left(\alpha_1\beta_1-1\right)\left(-\tau^2+\sqrt{\tau^2-4\sigma}\tau+4\sigma\right)}
{6\alpha_2\beta_1\epsilon\left(2+n_1^2\mu^2\tau^2\delta-4n_1^2\mu^2\delta\sigma\right)}\ ;\ P_1=\frac{n_1^2\mu^2\delta Q_1\left(\alpha_1\beta_1-1\right)\left(\tau^2-4\sigma\right)}
{3\alpha_2\beta_1\epsilon\left(2+n_1^2\mu^2\tau^2\delta-4n_1^2\mu^2\delta\sigma\right)}.\label{eq-66}
\end{align}
\begin{align}
P_2=P_3=0\ ,\ Q_0=-\frac{1}{2}\left(-\tau+\sqrt{\tau^2-4\sigma}\right)Q_1\ ;\ Q_1=Q_1\ ;\ \lambda=\pm\frac{\sqrt{2\alpha_1\beta_1-4n_1^2\mu^2\delta\sigma+n_1^2\mu^2\tau^2\delta}}
{\sqrt{\beta_1}\sqrt{2-4n_1^2\mu^2\delta\sigma+n_1^2\mu^2\tau^2\delta}}\ ;\ \mu=\mu.\label{eq-67}
\end{align}
Substituting $P_0\ ,\ P_1$ from the Eq. \eqref{eq-66}, $Q_0\ ,\ Q_1\ ,\ \lambda\ ,\ \mu$ from the Eq. \eqref{eq-67} into the Eq. \eqref{eq-53} yields the soliton-like solution when $\sigma\ne 0$ and $\tau^2-4\sigma>0$. The 2D plot for the imaginary part of the soliton is given in the Figure \ref{fig-6}. The 3D plot for the real and imaginary part of the soliton is given in the Figure \ref{fig-10} with the negative value of $\xi_7$ and with the positive value of $\xi_7$ in the Figure \ref{fig-11}.
\begin{align}
\Phi\left(x, t\right)=&-\left\{\frac{n_1^2\mu^2\delta\left(\alpha_1\beta_1-1\right)}
{3\alpha_2\beta_1\epsilon\left(2+n_1^2\mu^2\tau^2\delta-4n_1^2\mu^2\delta\sigma\right)}\right\}\nonumber\\
&\times\left\{
\frac{\left(-\tau^4+8\tau^2\sigma+\left(\tau^2-4\sigma\right)^{3/2}\tau-16\sigma^2\right)\tanh\left(\xi_7\right)
+\left(\tau^2-4\sigma\right)\left(\tau^2-\sqrt{\tau^2-4\sigma}\tau-4\sigma\right)}
{\left(\tanh\left(\xi_7\right)-1\right)\left(\tau^2-\sqrt{\tau^2-4\sigma}\tau-4\sigma\right)}
\right\}.\label{eq-68}
\end{align}
In the Eq. \eqref{eq-68} $\xi_7$ is given by the Eq. \eqref{eq-64}. And rational function solution when $\sigma\ne 0$ and $\tau^2-4\sigma<0$ by the Eq. \eqref{eq-65}.
\end{case}

\begin{case}
\label{sol11}
When
\begin{align}
P_0=\frac{\left(\alpha_1\beta_1-1\right)}{12\alpha_2\beta_1\epsilon\Lambda}
&\left\{
-12Q_0^5+24Q_1\tau Q_0^4+\left(-24Q_1^2\sigma+6\sqrt{\Lambda}-12Q_1^2\tau^2\right)Q_0^3
+\left(-12Q_1\tau\sqrt{\Lambda}+24\sigma\tau Q_1^3\right)Q_0^2\right.\nonumber\\
&\left.+\left(5Q_1^2\tau^2\sqrt{\Lambda}+10Q_1^2\sigma\sqrt{\Lambda}+4Q_1^4\sigma^2+Q_1^4\tau^4-8Q_1^4\tau^2\sigma\right)Q_0
-6Q_1^3\sigma\tau\sqrt{\Lambda}
\right\}.\label{eq-69}
\end{align}
\begin{align}
P_1=\frac{Q_1\left(\alpha_1\beta_1-1\right)\left(\sqrt{\Lambda}-Q_1^2\tau^2-2Q_1^2\sigma+6Q_0Q_1\tau-6Q_0^2\right)}
{12\alpha_2\beta_1\epsilon\sqrt{\Lambda}}\ ;\ P_2=P_3=0.\label{eq-70}
\end{align}
\begin{align}
Q_0=Q_0\ ;\ Q_1=Q_1\ ;\ \lambda=\pm\frac{\sqrt{2}\sqrt{\alpha_1\beta_1+1}}{2\sqrt{\beta_1}}\ ;\ \mu=\pm\frac{i\sqrt{2}Q_1}{\sqrt{\delta}n_1\Lambda^{1/4}}\ ;\ i=\sqrt{-1}.\label{eq-71}
\end{align}
In the Eqs. \eqref{eq-69}, \eqref{eq-70} and \eqref{eq-71} $\Lambda$ is given by,
\begin{align}
\Lambda=-12Q_0^4+24Q_0^3Q_1\tau+24Q_1^3Q_0\sigma\tau-12Q_0^2Q_1^2\tau^2
-24Q_0^2Q_1^2\sigma-8Q_1^4\tau^2\sigma+4Q_1^4\sigma^2+Q_1^4\tau^4.\label{eq-72}
\end{align}
Now substituting $P_0$ from the Eq. \eqref{eq-69}, $P_1$ from the Eq. \eqref{eq-70}, $Q_0\ ,\ Q_1\ ,\ \lambda\ ,\ \mu$ from the Eq. \eqref{eq-71} in the Eq. \eqref{eq-53} pertains the soliton-like solution when $\sigma\ne 0$ and $\tau^2-4\sigma>0$.
\begin{align}
\Phi\left(x, t\right)=\left\{\frac{P_0+P_1\exp\left[-\ln\left(
-\frac{\sqrt{\tau^2-4\sigma}}{2\sigma}\tanh\left\{\frac{\sqrt{\tau^2-4\sigma}}{2}\left(\mu\left(x-\lambda t\right)\right)+e\right\}
-\frac{\tau}{2\sigma}\right)\right]}{Q_0+Q_1\exp\left[-\ln\left(
-\frac{\sqrt{\tau^2-4\sigma}}{2\sigma}\tanh\left\{\frac{\sqrt{\tau^2-4\sigma}}{2}\left(\mu\left(x-\lambda t\right)\right)+e\right\}
-\frac{\tau}{2\sigma}\right)\right]}\right\}.\label{eq-73}
\end{align}
And singular periodic wave solution when $\sigma\ne 0$ and $\tau^2-4\sigma<0$.
\begin{align}
\Phi\left(x, t\right)=\left\{\frac{P_0+P_1\exp\left[-\ln\left(
-\frac{\sqrt{-\tau^2+4\sigma}}{2\sigma}\tan\left\{\frac{\sqrt{-\tau^2+4\sigma}}{2}\left(\mu\left(x-\lambda t\right)\right)+e\right\}
-\frac{\tau}{2\sigma}\right)\right]}{Q_0+Q_1\exp\left[-\ln\left(
-\frac{\sqrt{-\tau^2+4\sigma}}{2\sigma}\tan\left\{\frac{\sqrt{-\tau^2+4\sigma}}{2}\left(\mu\left(x-\lambda t\right)\right)+e\right\}
-\frac{\tau}{2\sigma}\right)\right]}\right\}.\label{eq-74}
\end{align}
\end{case}

\begin{case}
\label{sol12}
When
\begin{align}
P_0=-\frac{\left(\alpha_1\beta_1-1\right)}{12\alpha_2\beta_1\epsilon\Lambda}
&\left\{
12Q_0^5-24Q_1\tau Q_0^4+\left(24Q_1^2\sigma+6\sqrt{\Lambda}+12Q_1^2\tau^2\right)Q_0^3
+\left(-12Q_1\tau\sqrt{\Lambda}-24\sigma\tau Q_1^3\right)Q_0^2\right.\nonumber\\
&\left.+\left(5Q_1^2\tau^2\sqrt{\Lambda}+10Q_1^2\sigma\sqrt{\Lambda}-4Q_1^4\sigma^2-Q_1^4\tau^4+8Q_1^4\tau^2\sigma\right)Q_0
-6Q_1^3\sigma\tau\sqrt{\Lambda}
\right\}.\label{eq-75}
\end{align}
\begin{align}
P_1=\frac{Q_1\left(\alpha_1\beta_1-1\right)\left(\sqrt{\Lambda}+Q_1^2\tau^2+2Q_1^2\sigma-6Q_0Q_1\tau+6Q_0^2\right)}
{12\alpha_2\beta_1\epsilon\sqrt{\Lambda}}\ ;\ P_2=P_3=0.\label{eq-76}
\end{align}
\begin{align}
Q_0=Q_0\ ;\ Q_1=Q_1\ ;\ \lambda=\pm\frac{\sqrt{2}\sqrt{\alpha_1\beta_1+1}}{2\sqrt{\beta_1}}\ ;\ \mu=\pm\frac{\sqrt{2}Q_1}{\sqrt{\delta}n_1\Lambda^{1/4}}.\label{eq-77}
\end{align}
In the Eqs. \eqref{eq-75}, \eqref{eq-76} and \eqref{eq-77} $\Lambda$ is given by the Eq. \eqref{eq-72}. Substituting the $P_0$ from the Eq. \eqref{eq-75}, $P_1$ from the Eq. \eqref{eq-76}, $Q_0\ ,\ Q_1\ ,\ \lambda\ ,\ \mu$ from the Eq. \eqref{eq-77} gives the soliton-like solution given by the Eq. \eqref{eq-73} when $\sigma\ne 0$ and $\tau^2-4\sigma>0$ and singular periodic wave solution given by the Eq. \eqref{eq-74} when $\sigma\ne 0$ and $\tau^2-4\sigma<0$.
\end{case}
The solutions given in the Cases \ref{sol7} throught \ref{sol12} is with respect to the \textbf{Set 1} and \textbf{Set 2} of the Eqs. \eqref{eq-20} and \eqref{eq-21} respectively. Next the solution from the \textbf{Set 3} is computed by taking $\sigma=0$ in the differential equation Eq. \eqref{eq-19}. Hence the following differential equation is utilized.
\begin{align}
\frac{d\varphi\left(\xi\right)}{d\xi}=\exp\left(-\varphi\left(\xi\right)\right)
+\tau.\label{eq-78}
\end{align}
Now substituting the initial solution Eq. \eqref{eq-53} including Eq. \eqref{eq-78} in the Eq. \eqref{eq-25} results in the polynomial of $\exp\left(-\varphi\left(\xi\right)\right)$. Extracting the coefficent of each $\exp\left(-\varphi\left(\xi\right)\right)$ and it's powers gives the systems of eight algebraic equations. Solving the overdetermined systems of equations gives the following solutions of dispersive wave equation.

\begin{case}
\label{sol13}
When
\begin{align}
P_0=0\ ;\ P_1=-\frac{2Q_0n_1^2\mu^2\tau\delta\left(\alpha_1\beta_1-1\right)}
{\alpha_2\beta_1\epsilon\left(n_1^2\mu^2\tau^2\delta-2\right)}\ ;\ P_2=-\frac{2n_1^2\mu^2\delta\left(Q_0+\tau Q_1\right)\left(\alpha_1\beta_1-1\right)}{\alpha_2\beta_1\epsilon\left(n_1^2\mu^2\tau^2\delta-2\right)}.\label{eq-79}
\end{align}
\begin{align}
P_3=-\frac{2Q_1n_1^2\mu^2\delta\left(\alpha_1\beta_1-1\right)}
{\alpha_2\beta_1\epsilon\left(n_1^2\mu^2\tau^2\delta-2\right)}\ ;\ Q_0=Q_0\ ;\ Q_1=Q_1\ ;\ \lambda=\pm\sqrt{-\frac{2\alpha_1\beta_1-n_1^2\mu^2\tau^2\delta}{n_1^2\mu^2\tau^2\beta_1\delta-2\beta_1}}\ ;\ \mu=\mu.\label{eq-80}
\end{align}
Substituting $P_1\ ,\ P_2$ from the Eq. \eqref{eq-79}, $P_3\ ,\ Q_0\ ,\ Q_1\ ,\ \lambda\ ,\ \mu$ from the Eq. \eqref{eq-80} in the Eq. \eqref{eq-53} yields the exponential function solution when $\sigma=0$, $\tau\ne 0$ and $\tau^2-4\sigma>0$.
\begin{align}
\Phi\left(x, t\right)=-\frac{2n_1^2\mu^2\tau^2\delta\left(\alpha_1\beta_1-1\right)}
{\alpha_2\beta_1\epsilon\left(n_1^2\mu^2\tau^2\delta-2\right)}
\left\{\frac{\exp\left(\tau\left(\xi_8+e\right)\right)}
{\left(\exp\left(\tau\left(\xi_8+e\right)\right)-1\right)^2}\right\}.\label{eq-81}
\end{align}
In the Eq. \eqref{eq-81} $\xi_8=\left(\mu\left[x\mp\sqrt{\frac{-2\alpha_1\beta_1+\delta n_1^2\mu^2\tau^2}{\beta_1\left(-2+\delta n_1^2\mu^2\tau^2\right)}}t\right]\right)$.
\end{case}

\begin{case}
\label{sol14}
When
\begin{align}
P_0=\frac{Q_0n_1^2\mu^2\tau^2\delta\left(\alpha_1\beta_1-1\right)}
{3\alpha_2\beta_1\epsilon\left(n_1^2\mu^2\tau^2\delta+2\right)}\ ;\ P_1=\frac{n_1^2\mu^2\delta\tau\left(-6Q_0-Q_1\tau+Q_1\alpha_1\beta_1\tau+6Q_0\alpha_1\beta_1\right)}
{3\alpha_2\beta_1\epsilon\left(n_1^2\mu^2\tau^2\delta+2\right)}.\label{eq-82}
\end{align}
\begin{align}
P_2=\frac{2n_1^2\mu^2\delta\left(Q_0+Q_1\tau\right)\left(\alpha_1\beta_1-1\right)}
{\alpha_2\beta_1\epsilon\left(n_1^2\mu^2\tau^2\delta+2\right)}\ ;\ P_3=\frac{2Q_1n_1^2\mu^2\delta\left(\alpha_1\beta_1-1\right)}
{\alpha_2\beta_1\epsilon\left(n_1^2\mu^2\tau^2\delta+2\right)}\ ;\ Q_0=Q_0\ ;\ Q_1=Q_1.\label{eq-83}
\end{align}
\begin{align}
\lambda=\pm\sqrt{-\frac{-2\alpha_1\beta_1-n_1^2\mu^2\tau^2\delta}{n_1^2\mu^2\tau^2\beta_1\delta+2\beta_1}}\ ;\ \mu=\mu.\label{eq-84}
\end{align}
Substituting $P_0\ ,\ P_1$ from the Eq. \eqref{eq-82}, $P_2\ ,\ P_3\ ,\ Q_0\ ,\ Q_1$ from the Eq. \eqref{eq-83}, $\lambda\ ,\ \mu$ from the Eq. \eqref{eq-84} ito the Eq. \eqref{eq-53} results in the exponential function solution when $\sigma=0$, $\tau\ne 0$ and $\tau^2-4\sigma>0$.
\begin{align}
\Phi\left(x, t\right)=\frac{\left(\alpha_1\beta_1-1\right)n_1^2\mu^2\tau^2\delta}
{3\alpha_2\beta_1\epsilon\left(n_1^2\mu^2\tau^2\delta+2\right)}
\left\{\frac{\left(\exp\left(\tau\left(\xi_9+e\right)\right)\right)^2
+4\exp\left(\tau\left(\xi_9+e\right)\right)+1}
{\left(\exp\left(\tau\left(\xi_9+e\right)\right)-1\right)^2}\right\}.\label{eq-85}
\end{align}
In the Eq. \eqref{eq-85} $\xi_9=\left(\mu\left[x\mp\sqrt{\frac{2\alpha_1\beta_1+\delta n_1^2\mu^2\tau^2}{\beta_1\left(2+\delta n_1^2\mu^2\tau^2\right)}}t\right]\right)$.
\end{case}

\begin{case}
\label{sol15}
When
\begin{align}
P_0=\frac{Q_0\left(-\lambda^2+\alpha_1\right)}{3\alpha_2\epsilon}\ ;\ P_1=\frac{Q_1\left(-\lambda^2+\alpha_1\right)}{3\alpha_2\epsilon}\ ;\ P_2=P_3=0.\label{eq-86}
\end{align}
\begin{align}
Q_0=Q_0\ ;\ Q_1=Q_1\ ;\ \lambda=\lambda\ ;\ \mu=\mu.\label{eq-87}
\end{align}
Now substituting $P_0\ ,\ P_1$ from the Eq. \eqref{eq-86}, $Q_0\ ,\ Q_1\ ,\ \lambda\ ,\ \mu$ from the Eq. \eqref{eq-87} into the Eq. \eqref{eq-53} yields the rational solution when $\sigma=0$, $\tau\ne 0$ and $\tau^2-4\sigma>0$.
\begin{align}
\Phi\left(x, t\right)=-\frac{\lambda^2-\alpha_1}{3\alpha_2\epsilon}.\label{eq-88}
\end{align}
\end{case}
The condition given in the \textbf{Set 4} is $\tau^2-4\sigma=0$. Therefore $\tau=\pm 2\sqrt{\sigma}$ (or) $\sigma=\frac{\tau^2}{4}$. So the differential equation in the Eq. \eqref{eq-19} changes to,
\begin{align}
\frac{d\varphi\left(\xi\right)}{d\xi}=\exp\left(-\varphi\left(\xi\right)\right)
+\sigma\exp\left(\varphi\left(\xi\right)\right)\pm 2\sqrt{\sigma}.\label{eq-89}
\end{align}
(or)
\begin{align}
\frac{d\varphi\left(\xi\right)}{d\xi}=\exp\left(-\varphi\left(\xi\right)\right)
+\frac{\tau^2}{4}\exp\left(\varphi\left(\xi\right)\right)+\tau.\label{eq-90}
\end{align}
Now substituting the Eq. \eqref{eq-53} with the Eq. \eqref{eq-89} (or) Eq. \eqref{eq-90} in the Eq. \eqref{eq-25} results in the polynomial $\exp\left(-i\varphi\left(\xi\right)\right)\ ;\ i=0-7$. Next collecting the coefficent of each exponential gives the nine algebraic equations systems. For the \textbf{Set 5} the given condition is $\sigma=\tau=0$ so the Eq. \eqref{eq-19} reduces into,
\begin{align}
\frac{d\varphi\left(\xi\right)}{d\xi}=\exp\left(-\varphi\left(\xi\right)\right)
.\label{eq-91}
\end{align}
To get the solution from the \textbf{Set 5} the initial solution Eq. \eqref{eq-53} is substituted in the Eq. \eqref{eq-25} along with the Eq. \eqref{eq-91}. This again gives the polynomial in $\exp\left(-i\varphi\left(\xi\right)\right)\ ;\ i=0-7$, extracting the coefficent yields nine systems of algebraic equations. Both the \textbf{Set 4} and \textbf{Set 5} gives the rational function solutions which are not reported in this work. We have the following theorem which gives about the number of solutions and generalization.

\begin{Thm}
\label{thm-1}
Let $M$ and $N$ be the non-zero positive integers. In the modified exponential function method if the balancing principle relation is given by the linear equation $N=M+2$, then there are $M+7$ (or) $N+5$ algebraic equations in the overdetermined systems with $2\left(M+3\right)$ (or) $2\left(N+1\right)$ unknowns (including the $\lambda$ and $\mu$ defined in the wave transformation).
\end{Thm}
\begin{proof}
The linear equation $N=M+2$ have infinitely many integer solutions. In this work we studied by taking $M=1$ and so $N=3$. Suppose if $M=2$ then $N=4$ so the initially assumed solution from the Eq. \eqref{eq-19} to the Eq. \eqref{eq-25} takes the following form.
\begin{align}
u\left(\xi\right)=\frac{P_0+P_1\exp\left(-\varphi\left(\xi\right)\right)
+P_2\exp\left(-2\varphi\left(\xi\right)\right)+P_3\exp\left(-3\varphi\left(\xi\right)\right)
+P_4\exp\left(-4\varphi\left(\xi\right)\right)}
{Q_0+Q_1\exp\left(-\varphi\left(\xi\right)\right)+Q_2\exp\left(-2\varphi\left(\xi\right)\right)}.\label{eq-92}
\end{align}
Now substituting Eq. \eqref{eq-92} along with Eq. \eqref{eq-19} (or) Eq. \eqref{eq-78} (or) Eq. \eqref{eq-89} (or) Eq. \eqref{eq-90} (or) Eq. \eqref{eq-91} in the ordinary differential equation Eq. \eqref{eq-25} results in the polynomial of $\exp\left(-i\varphi\left(\xi\right)\right)\ ;\ i=0-8$. Therefore the each coefficent equating to zero results in the systems of nine algebraic equations. Next in the linear equation $N=M+2$ taking $M=3$ gives $N=5$. Following the aforementioned procedure yields the systems of ten algebraic equations. Therefore continuing in the same way by the method of induction completes the proof of the theorem.
\end{proof}
\begin{Rem}
\label{rem-1}
In practical, in the linear equation $N=M+2$ when the $M$ value is greater than $2$, then the computer algerbra software keeps on executing without returning the results. Thus the software gets choked when the integer value $M$ is large.
\end{Rem}

\section{Graphical representations, results and interpretations}
\label{resu}

For the extended sine-Gordon method the two dimensional graph for the solitons in the Eqs. \eqref{eq-43} and \eqref{eq-48} are drawn by taking the Lam$\acute{e}$'s coefficent $\lambda_1=1.50$, $\mu_1=2.50$ so $n_1=\frac{3}{16}$, the constitutive constants $\nu_1=2$, $\nu_2=3$, $\nu_4=5$ hence $\kappa_1=\frac{82}{3}$, $\kappa_3=\frac{35}{2}$, $\kappa_5=\frac{23}{2}$, $\kappa_6=\frac{95}{12}$ and $c_1=\frac{95}{32}$, $c_2=\frac{38065}{1152}$. Next $\rho=3$, $c=4$ so $\beta_1=\frac{96}{5}$, $\alpha_1=\frac{95}{768}$, $\alpha_2=\frac{38065}{55296}$ and the small parameters are $\delta=2.5$, $\epsilon=3.5$. For the 2D graphs the time variable $t=1$ is taken and the wave number $\mu$ varied from $0.25$ to $1.25$ with the step of $0.25$ and drawn in the domain $-5\le x\le 5$. The three dimensional graphs for the solitons Eqs. \eqref{eq-35}, \eqref{eq-36} and \eqref{eq-51} is drawn by taking the numerical values $\lambda_1=0.75$, $\mu_1=1.25$ so the Lam$\acute{e}$'s coefficent $n_1=\frac{3}{16}$, constitutive constants $\nu_1=1$, $\nu_2=2$, $\nu_4=4$ hence $\kappa_1=16$, $\kappa_3=10.75$, $\kappa_5=6.75$, $\kappa_6=\frac{37}{8}$. Next $\rho=2.5$, $c=3.5$ so that $c_1=\frac{95}{64}$, $c_2=\frac{3719}{192}$, $\alpha_1=\frac{19}{196}$, $\alpha_2=\frac{3719}{5880}$, $\beta_1=\frac{49}{2}$ and the small parameters $\delta=1$, $\epsilon=2$. By taking the wave number $\mu=2.25$ the 3D graphs are drawn in the domain $-15\le x\ ,\ t\ \le 15$.

For the modified exponential function method the two dimensional graphs for the solitons Eqs. \eqref{eq-56}, \eqref{eq-64} and \eqref{eq-68} is drawn with same numerical values of the sine-Gordon method along with $\tau=2.50$, $\sigma=2.50$, integration constant $e=2$, constant coefficent $Q_0=Q_1=2$. At the time $t=1.50$ by varying the frequency $\lambda$ from $1$ to $5$ in the step of $1$ within the domain $0\le x\le 10$. The three dimensional graphs are drawn with the same numerical values of 3D numerical values of sine-Gordon method in addition $\tau=1.25$, $\sigma=2.25$, integration constant $e=5$, $Q_0=2$, $Q_1=3$ at the frequency $\lambda=2$ within the domain $-15\le x\, t\ \le 15$.

With respect to the extended sine-Gordon method the coefficent $B_2$ in the Cases \ref{sol3} to \ref{sol6} are imaginary therefore the solitons reported in the Eqs. \eqref{eq-39}, \eqref{eq-43}, \eqref{eq-47} and \eqref{eq-51} are the complex structured solitons. When the wave number increasing the solitons travelling with the higher wavelength.

In the modified exponential function method the wave number $\mu$ is complex valued and hence gives the complex structured solitons and singular periodic wave solutions in the Cases \ref{sol8} and \ref{sol11}. When the frequency $\lambda$ increases the travelling waves have high amplitude.

\section{Conclusion}
\label{conc}

In this research paper the nonlinear dispersive wave Eq. \eqref{eq-1} defined by the nonlinear partial differential equation in the cylindrical elastic rod having Murnaghan's materials given by the Eqs. \eqref{eq-2}-\eqref{eq-5} is solved for the unknown function (dispersive wave) $\Phi\left(x, t\right)$ of the Eq. \eqref{eq-1} using extended sine-Gordon equation method and modified exponential function method. Extended sine-Gordon method gives the topological (or) dark soliton, compound topological-non-topological (bright) soliton and singular solitons. However the extended sine-Gordon method do not give the non-topological (bright) soliton. With respect to extended sine-Gordon method we have reported all the solutions obtained. The selective 2D and 3D graphs are drawn to show the solitonic structures.

The modified exponential function method have five basic (or) initial solutions given in the \textbf{Set 1} to \textbf{Set 5} and five auxiliary differential equations. The \textbf{Set 1} and \textbf{Set 2} is for the Eq. \eqref{eq-19}, the \textbf{Set 3} is for the Eq. \eqref{eq-78}, the \textbf{Set 4} is for the Eqs. \eqref{eq-89} (or) \eqref{eq-90}, the \textbf{Set 5} is for the Eq. \eqref{eq-91}, and in this work we reported only the \textbf{Set 1}-\textbf{Set 3} as the \textbf{Set 4} and \textbf{Set 5} gives the rational solutions we did not included in this work. This method gave the soliton-like, singular periodic wave and exponential function solutions. The Figures \ref{fig-8} and \ref{fig-9}, Figures \ref{fig-10} and \ref{fig-11} shows the variation of negative and positive value of $\xi$ in the soliton. The existence of number of algebraic equations and number of unknowns are proved in the Theorem \ref{thm-1}. To the resources collected by the authors the dispersive equation given by the Eq. \eqref{eq-1} is not studied using two integral methods previously and therefore the solutions are first time appearing in this communication work.

\end{document}